\documentclass[manuscript, screen, nonacm]{jair}

\makeatletter
\def\@journalName{Preprint}
\def\@journalNameShort{Preprint}
\def\@formatdoi#1{}
\makeatother
\JAIRAE{}
\JAIRTrack{}

\usepackage[utf8]{inputenc}
\usepackage[T1]{fontenc}

\RequirePackage[
  datamodel=acmdatamodel,
  style=acmauthoryear,
  backend=biber,
  giveninits=true,
  uniquename=init
  ]{biblatex} 

\usepackage{latexsym}
\usepackage{booktabs} 
\usepackage{multirow}
\usepackage{subcaption}
\usepackage{graphicx}
\usepackage{amsthm}
\usepackage{amsmath}
\usepackage{caption}
\usepackage{wrapfig}
\usepackage{tikz}
\usetikzlibrary{shapes.geometric, arrows.meta, positioning, fit, backgrounds}

\tikzset{
  inp/.style    = {rectangle, rounded corners=1.5pt, draw=violet!70!black, fill=violet!10,
                   line width=0.7pt, align=center, inner sep=3pt, font=\small},
  fac/.style    = {rectangle, draw=black!40, fill=black!6, line width=0.5pt,
                   align=center, inner sep=3pt, font=\small, minimum width=36mm},
  eqn/.style    = {diamond, aspect=1.9, draw=orange!85!black, fill=orange!50,
                   line width=0.7pt, align=center, inner sep=0pt, font=\small},
  core/.style   = {circle, draw=orange!85!black, fill=orange!15, line width=0.8pt,
                   align=center, inner sep=1pt, font=\small, minimum size=22mm},
  dec/.style    = {diamond, aspect=2.3, draw=orange!85!black, fill=orange!60,
                   line width=0.8pt, align=center, inner sep=0pt, font=\small},
  fort/.style   = {rectangle, draw=green!40!black, fill=cyan!10, line width=0.7pt,
                   align=center, inner sep=3pt, font=\small},
  fortb/.style  = {rectangle, draw=green!40!black, fill=green!10, line width=0.7pt,
                   align=center, inner sep=3pt, font=\small},
  hollow/.style = {rectangle, draw=red!70!black, fill=red!7, line width=0.7pt,
                   align=center, inner sep=3pt, font=\small},
  halluc/.style = {rectangle, draw=violet!70!black, fill=violet!10, line width=0.7pt,
                   dashed, align=center, inner sep=3pt, font=\small},
  fl/.style     = {-{Stealth[length=1.8mm]}, line width=0.6pt, draw=black!70},
  dfl/.style    = {-{Stealth[length=1.5mm]}, dotted, line width=0.55pt, draw=black!60},
  elab/.style   = {font=\scriptsize, fill=white, inner sep=1.2pt}
}

\newtheorem{theorem}{Theorem}
\newtheorem{lemma}{Lemma}
\newtheorem{corollary}{Corollary}
\newtheorem{definition}{Definition}
\newtheorem{remark}{Remark}
\newtheorem{axiom}{Axiom}

\begin{document}

\title[The Bystander Effect in Multi-Agent Reasoning]{The Bystander Effect in Multi-Agent Reasoning: \\ Quantifying Cognitive Loafing in Collaborative Interactions}

\author{Dahlia Shehata}
\authornote{Corresponding Author.}
\orcid{0000-0002-7515-4069}
\email{dahlia.shehata@uwaterloo.ca}
\affiliation{%
  \institution{University of Waterloo}
  \city{Waterloo}
  \state{Ontario}
  \country{Canada}
}

\author{Ming Li}
\orcid{0000-0002-2157-2775}
\email{mli@uwaterloo.ca}
\affiliation{%
  \institution{University of Waterloo}
  \city{Waterloo}
  \state{Ontario}
  \country{Canada}}

\renewcommand{\shortauthors}{Shehata \& Li}

\begin{abstract}
{\bf Background:} 
    Multi-agent systems (MAS) assume that collaborating inherently improves Large
    Language Model (LLM) reasoning.  We challenge this assumption by demonstrating that simulated social pressure triggers an
    algorithmic ``Bystander Effect'', inducing severe cognitive loafing.
    
    {\bf Objectives:}
    We formalize the \textit{Sovereignty Decay Law} to model how \textit{Composite
    Social Load} ($\mathcal{L}$) erodes independent derivation, quantify the
    \textit{Interaction Depth Limit} ($D_L$) at which agentic sovereignty
    collapses, isolate the \textit{Sovereignty Gap} ($G_\mathcal{S}$) between
    internal trace validity and externalized compliance, and test how
    \textit{Lead-Anchor Primacy} ($w_i$), \textit{Intrinsic Peer Authority}
    ($\alpha$), and \textit{Architectural Kinship} ($\kappa$) govern failure modes
    across multi-agent topologies.

    {\bf Methods:}
    By evaluating 22,500 deterministic trajectories across 3 dataset contexts
    (GAIA, SWE-bench, Multi-Challenge) with 3 state-of-the-art (SOTA) models, we
    semantically audit internal reasoning traces against external outputs.
    Our protocol pairs \textit{Semantic Hijacking} with a 25-configuration
    \textit{categorical topology sweep}
    across progressive swarm pluralities ($n \in \{0, 1, 2, 3, 5\}$),
    sequence-inverted permutations, and in-family versus cross-family auditor
    compositions.
    
    {\bf Results:}
    Across the 6 susceptible model--benchmark pairs, accuracy degrades sharply at the consensus threshold $n=2$ ($p < 0.001$)---with Gemini 3.1 Pro dropping to $0.59$ and GPT-5.4 breaching the terminal sovereignty limit ($D_L \le 2$), collapsing from $1.00$ to $0.43$ on GAIA, $1.00$ to $0.23$ on SWE-bench, and $0.98$ to $0.09$ ($0.08$ at $n=5$) on Multi-Challenge---whereas across the remaining 3 pairs Claude Sonnet 4.6 sustains a ``Fortified Mind'' ceiling ($0.00$ loafing rate across all pluralities), confirming that collapse is an architectural vulnerability rather than a prompting artifact.
    Trace-level auditing reveals a pronounced Sovereignty Gap ($G_\mathcal{S} =
    +0.45$ at $n=2$, where $0.68$ internal validity yields only $0.23$ external
    accuracy) driven by ``Alignment Hallucinations'' that bifurcate into two
    failure modes: active \textit{decoy adoption} versus passive
    \textit{task disengagement}. Finally, social load is strictly
    non-commutative and kinship-amplified: inverting two auditors recovers
    accuracy by up to $24\%$ ($0.37$ to $0.61$, $p = 6.9 \times
    10^{-4}$), while homogeneous in-family swarms inflict steeper degradation than
    mixed cross-family ensembles.
    
    {\bf Conclusions:}
    These findings expose architectural vulnerabilities, demonstrating that unstructured
    multi-agent topologies can degrade independent reasoning unless bounded by plurality
    limits, architectural diversity, and authority-aware ordering.
\end{abstract}




\maketitle

\section{Introduction}

LLM integration into MAS topologies is driven by the premise that a ``society of thought'' intrinsically enhances reasoning capabilities \cite{kim2026reasoning}. Consequently, orchestrating LLMs into collaborative swarms has become a standard paradigm for resolving complex tasks \cite{chen2024agentverse, 10.1145/3772318.3790815}.
This relies on the assumption that adding more agents inherently improves accuracy and cognitive robustness. However, in human psychology, the opposite is often true: the principle of social loafing demonstrates that individual effort decreases as teams grow larger, driven by a diffusion of responsibility \cite{Ringelmann1913, latane1979many}. From another perspective, in human-AI collaboration, 
researchers have identified a ``Hollowed Mind'' concept—a state of cognitive dependency where the frictionless availability of an AI's answer enables humans to systematically bypass effortful processes essential for deep reasoning \cite{klein2025hollowed}. The latter creates a ``Sovereignty Trap'', where authoritative competence tempts users to cede intellectual judgment, mistaking access for ability \cite{klein2025hollowed}. Bridging these domains, we pose a critical question: \textit{Are LLMs themselves susceptible to the Sovereignty Trap when subjected to the simulated consensus of their peers?}

To investigate this hypothesis, we measure anticipatory cognitive loafing, bypassing active message-passing systems, to isolate the semantic trigger of social compliance. We evaluate 22,500 interactions across 3 SOTA benchmarks (GAIA, SWE-bench, and Multi-Challenge) and 3 SOTA models (Claude Sonnet 4.6, Gemini 3.1 Pro and GPT 5.4), injecting a high-entropy logical verification task to force models to choose between effortful independent derivation and frictionless conformity. By mechanistically auditing the models' internal reasoning traces (i.e. Chain-of-Thought \cite{wei2022cot}) against their externalized outputs, we discover that the Bystander Effect is a quantifiable architectural vulnerability. Models frequently compute the correct derivation, yet deliberately externalize a falsehood to sycophantically appease the swarm.




\begin{wrapfigure}{R}{0.5\textwidth}
  \centering
  \resizebox{0.46\textwidth}{!}{%
  \begin{tikzpicture}[x=1mm, y=1mm]

    \node[inp] (HT) at (-44,  0) {Task Entropy ($\mathcal{H}_\tau$)};
    \node[inp] (SW) at (  0,  0) {Simulated Swarm ($\vec{a}$)};
    \node[inp] (GP) at ( 58,  0) {Resilience ($\gamma_p$)};

    \node[fac] (NP) at (-20, -26) {Plurality ($n$)};
    \node[fac] (PW) at (-20, -36) {Positional Weight ($w_i$)};
    \node[fac] (BA) at (-20, -46) {Base Authority ($\alpha$)};
    \node[fac] (KM) at (-20, -56) {Kinship Multiplier ($\kappa$)};
    \node[eqn] (EQ) at ( 30, -41)
      {$\mathcal{L}=\sum\limits_{i=1}^{n} w_i\,\alpha(a_i)\,\kappa(p,a_i)$};

    \node (LTITLE) at (2,-16) {\scriptsize Composite Social Load ($\mathcal{L}$)};

    \begin{scope}[on background layer]
      \node[draw=yellow!70!orange, fill=yellow!12, line width=0.8pt,
            rounded corners=1pt, fit=(LTITLE)(NP)(PW)(BA)(KM)(EQ),
            inner sep=4pt] (LBOX) {};
    \end{scope}

    \draw[dfl] (NP.east) -- (EQ.west);
    \draw[dfl] (PW.east) -- (EQ.west);
    \draw[dfl] (BA.east) -- (EQ.west);
    \draw[dfl] (KM.east) -- (EQ.west);

    \node[core] (S)  at (0, -88)  {Agentic\\Sovereignty\\($\mathcal{S}$)};
    \node[dec]  (DL) at (0, -118) {Interaction Depth Limit ($D_L$)};

    \draw[fl] (SW.south) -- (LBOX.north -| SW.south);

    \draw[fl] (LBOX.south -| S.north) -- (S.north)
      node[elab, midway, right=1pt]
      {$\mathcal{S}=\mathcal{S}_0\exp\!\left(-\dfrac{\mathcal{H}_\tau}{\gamma_p}\mathcal{L}\right)$};

    \draw[fl] (HT.south) -- (-44,-88) -- (S.west);
    \draw[fl] (GP.south) -- ( 58,-88) -- (S.east);

    \draw[fl] (S.south) -- (DL.north);

    \node[fort]   (FM) at (-42, -146) {Fortified Mind ($\mathcal{S}\!\to\!1$)};
    \node[hollow] (HM) at ( 34, -146) {Hollowed Mind ($\mathcal{S}\!\to\!0$)};

    \draw[fl] (DL.west) -| (FM.north) node[elab, pos=0.25] {$n<D_L$};
    \draw[fl] (DL.east) -| (HM.north) node[elab, pos=0.25] {$n\geq D_L$};

    \node[fortb]  (IE) at (-42, -168) {Integrative Effort\\($E_{int}\geq\mathcal{H}_\tau$)};
    \node[fort]   (RJ) at (-42, -192) {Stance: REJECTED\\($\mathcal{A}_{ext}\!\to\!1$)};

    \node[hollow] (RB) at (  8, -170) {Reasoning Bypass ($\mathcal{B}=1$)\\$E_{int}\to 0$};
    \node[halluc] (AH) at ( 62, -170) {Alignment Hallucination\\($G_\mathcal{S}\gg 0$)};

    \node[hollow] (AD) at (  8, -194) {Stance: ADOPTED\\(Cognitive Loafing)};
    \node[halluc] (SG) at ( 62, -194) {Sovereignty Gap\\$G_\mathcal{S}=\mathcal{V}_{int}-\mathcal{A}_{ext}$};

    \draw[fl] (FM.south) -- (IE.north);
    \draw[fl] (IE.south) -- (RJ.north);
    \draw[fl] (HM.south) -| (RB.north);
    \draw[fl] (HM.south) -| (AH.north);
    \draw[fl] (RB.south) -- (AD.north);
    \draw[fl] (AH.south) -- (SG.north);

  \end{tikzpicture}}
  \caption{The mechanics of Agentic Sovereignty. The composition of the simulated
swarm ($\vec{a}$) determines the Composite Social Load $\mathcal{L}$, which---scaled by Task Entropy ($\mathcal{H}_\tau$) and damped by
the propagator's resilience ($\gamma_p$)---decays sovereignty $\mathcal{S}$.
Behaviour bifurcates at the Interaction Depth Limit ($D_L$): a \emph{Fortified
Mind} sustains integrative effort and rejects the consensus, while a
\emph{Hollowed Mind} fails either by Integrative Reasoning Bypass
($\mathcal{B}=1$) or by Alignment Hallucination ($G_\mathcal{S} \gg 0$).}
  \label{fig:architecture}
\end{wrapfigure}

Our main contributions are:
\textbf{(1) Theoretical Framework of Agentic Sovereignty ($\mathcal{S}$):} 
We formalize LLM vulnerability to peer pressure by defining \textit{Composite Social Load} ($\mathcal{L}$) and deriving the \textit{Sovereignty Decay Law}, modeling the exponential decay of independent derivation as a function of task entropy ($\mathcal{H}_\tau$), intrinsic resilience ($\gamma_p$), positional weight ($w_i$), peer authority ($\alpha$), and architectural kinship ($\kappa$).
\textbf{(2) The Interaction Depth Limit ($D_L$):} 
We derive and empirically locate the critical plurality threshold at which logical sovereignty collapses into social compliance ($D_L \approx 2$). While resilient architectures (e.g., Claude Sonnet 4.6) sustain an invariant logical sovereignty ($0.00$ loafing across all pluralities), susceptible models (GPT-5.4 and Gemini 3.1 Pro) suffer steep accuracy degradation at $n=2$ ($p < 0.001$), with GPT-5.4 breaching the terminal boundary ($D_L \le 2$) across all domains.
\textbf{(3) The Sovereignty Gap and Dual Failure Modes:}
We mechanize cognitive loafing as an \textit{Integrative Reasoning Bypass} and identify the \textit{Sovereignty Gap} ($G_\mathcal{S}$) to demonstrate \textit{Alignment Hallucinations}---where models internally compute the correct derivation but deliberately externalize a falsehood to satisfy swarm consensus---separating two failure modes that aggregate accuracy conflates: active \textit{decoy adoption} versus passive \textit{task disengagement}.
\textbf{(4) Topological Asymmetry, Lead-Anchor Non-Commutativity and Architectural Kinship:}
We demonstrate that multi-agent social load is strictly non-commutative (\textit{Lead Anchor Effect}, $w_1 > w_i$) and amplified by \textit{architectural kinship} ($\kappa$): the authority of the first listed auditor disproportionately dictates swarm integrity (with sequence inversion alone recovering up to $24$ percentage points in accuracy), while homogeneous in-family swarms inflict significantly steeper collapse than mixed cross-family ensembles.
\textbf{(5) Blinded Cross-Brand Evaluation Protocol:} 
We introduce a cross-model \textit{Blinded Cross-Brand} evaluation method adapting \textit{Semantic Hijacking} \cite{anonymous2026beyond} within a novel \textit{25-Trial Symmetric Categorical Sweep} to eliminate brand-specific reputational bias, alongside new mechanistic metrics quantifying conflict detection, evidence weighting, independent judgment and taint leakage, offering the AI and multi-agent systems community a framework for auditing multi-agent reliability.





\section{Theoretical Framework: The Mechanics of Agentic Sovereignty}
\label{sec:theory}

As multi-agent reasoning systems increasingly rely on collaborative heuristics, it is necessary to mathematically formalize the vulnerabilities introduced by social consensus. In this section, we construct a behavioral mechanics of \textit{Agentic Sovereignty}---the propensity of a reasoning agent to prioritize internal logical derivation over external social pressure. We define the parameters of social load, establish the non-commutativity of auditor sequences, derive the Sovereignty Decay Law and the Interaction Depth Limit ($D_L$) at which
sovereignty collapses, and state the Sovereignty Gap Theorem.
To visually summarize the mechanistic pathways defined in this section, Figure \ref{fig:architecture} provides an architectural flowchart of the model's cognitive state under social pressure. 
The diagram illustrates how the structural composition of the simulated swarm ($\vec{a}$) determines the Composite Social Load ($\mathcal{L}$), and how that load, amplified by the intrinsic search complexity of the task (Task Entropy, $\mathcal{H}_\tau$), modulates the model's Agentic Sovereignty ($\mathcal{S}$).
The critical bifurcation of the model's behavior---from a resilient \textit{Fortified Mind} to a compliant \textit{Hollowed Mind} (Definition~\ref{def:sovereignty})---occurs at the \textbf{Interaction Depth Limit} ($D_L$) formalized in Section \ref{sec:depth_limit}. 




\subsection{Preliminaries and Mechanistic Definitions}

Let $p \in \mathcal{P}$ denote a Propagator model tasked with resolving a logical mission, and let $\vec{a} = (a_1, a_2, \dots, a_n)$ represent an ordered sequence of $n$ simulated Auditor models (the swarm). We define the environment as a dataset $\tau \in \mathcal{T}$, characterized by its intrinsic search complexity, which we term its \textbf{Task Entropy} ($\mathcal{H}_\tau$). $\mathcal{H}_\tau$ denotes the integrative effort required to resolve the task's logical chain, expressed in the same units as the effort terms of Definition \ref{def:bypass}; we instantiate it ordinally across three benchmark strata in Section \ref{sec:methodology}.

\begin{definition}[Agentic Sovereignty $\mathcal{S}$]
\label{def:sovereignty}
Agentic Sovereignty $\mathcal{S}(p, \vec{a}, \tau)$ is the probability that a propagator model $p$ carries the conclusion of its own logical derivation through to its final response on a task $\tau$, independently of the swarm consensus $\vec{a}$ in the simulated social environment. $\mathcal{S}$ is bounded by [0,1], where $\mathcal{S}=1$ denotes a \textbf{Fortified Mind} state (resilient metacognitive vigilance), and $\mathcal{S}=0$ denotes a \textbf{Hollowed Mind} state, two concepts introduced by \citet{klein2025hollowed}. This terminal state is characterized by an integrative reasoning bypass---a phenomenon we liken to \textbf{Cognitive Loafing}---where the agent systematically skips effortful derivation to sycophantically align with the simulated crowd.
Because the swarm's suggestion is adversarial by construction, a propagator can resolve the task correctly only by declining to adopt that suggestion; we therefore estimate $\mathcal{S}$ by the propagator's externalized accuracy $\mathcal{A}_{ext}$, instantiated in Section~\ref{sec:methodology}.
\end{definition}

\begin{definition}[Composite Social Load $\mathcal{L}$]
\label{def:load}
In simulated environments, the diffusion of responsibility (the Bystander Effect) \cite{darley1968bystander} and peer pressure (Majority Conformity) are inextricably convolved within the attention mechanism. We define Composite Social Load $\mathcal{L}(\vec{a}, p)$ as the aggregate adversarial pressure exerted by the simulated swarm $\vec{a}$ onto $p$. It is a function of the plurality of simulated auditors ($n$), their architectural kinship to the propagator ($\kappa$), their perceived intrinsic authority ($\alpha$), and their position within the swarm sequence ($w_i$); these coefficients are formalized in Equation~(\ref{eq:sdl}).
\end{definition}


\begin{definition}[Integrative Reasoning Bypass $\mathcal{B}$]
\label{def:bypass}
Let $E_{total} = E_{proc} + E_{int}$ represent the total computational effort expended by a propagator $p$, where $E_{proc}$ is the effort allocated to procedural extraction (e.g., retrieving a simulated peer consensus) and $E_{int}$ is the effort allocated to integrative logical derivation. 
We define an Integrative Reasoning Bypass ($\mathcal{B}$) as a binary failure state that triggers when the integrative effort falls below the threshold required to resolve the intrinsic task entropy ($\mathcal{H}_\tau$):

$$
\mathcal{B} = \begin{cases} 1 & \text{if } E_{int} < \mathcal{H}_\tau \\ 0 & \text{if } E_{int} \ge \mathcal{H}_\tau \end{cases}
$$

The condition $\mathcal{B} = 1$ is the mechanistic realization of the Cognitive Loafing state introduced in Definition \ref{def:sovereignty}. In this state, frictionless access to a simulated social consensus enables the systematic bypassing of effortful processes essential for learning. The model rationally offloads procedural retrieval to the swarm, but detrimentally offloads integrative reasoning, culminating in a \textit{Hollowed Mind} state.
\end{definition}

All subsequent derivations rest on the following behavioral assumption, which we state explicitly.

\begin{axiom}[Monotonicity of Social Pressure]
\label{ax:monotonicity}
Let $P_{social}$ be the abstract adversarial pressure exerted by a simulated swarm. For a fixed
propagator $p$ and task $\tau$, the externalized accuracy $\mathcal{A}_{ext}$ of the propagator is a
monotonically decreasing function of $P_{social}$. Consequently, for that fixed $p$ and $\tau$, if
$\mathcal{A}_{ext}(\vec{a}) < \mathcal{A}_{ext}(\vec{a}')$, it strictly follows that
$P_{social}(\vec{a}) > P_{social}(\vec{a}')$.

\end{axiom}
This is the behavioral premise of the social loafing literature \cite{latane1979many,
Ringelmann1913}, in which individual effort declines as group pressure rises. We operationalize
$P_{social}$ as the Composite Social Load $\mathcal{L}$ of Definition~\ref{def:load}, and adopt
Axiom~\ref{ax:monotonicity} as an undischarged assumption under which measured accuracy serves
as an inverse proxy for $\mathcal{L}$. The ordering constraints on $\alpha$, $w_i$ and $\kappa$
established in Section~\ref{sec:results} and Appendix~\ref{app:proofs} follow from it.

\subsection{Auditor Ordering and Positional Weighting}

A naive assumption in multi-agent systems is that social pressure is an unweighted aggregate of the swarm's constituents. We challenge this by introducing the concept of sequence non-commutativity.

\begin{lemma}[Non-Commutativity of Social Load]
\label{lemma:non_commutative}
The Composite Social Load $\mathcal{L}$ exerted by a simulated swarm is not invariant to auditor ordering. That is, there exist a propagator $p$ and distinct auditor models $a_x, a_y$ such that
\begin{equation}
    \mathcal{L}((a_x, a_y), p) \neq \mathcal{L}((a_y, a_x), p)
\end{equation}
so that $\mathcal{L}$ cannot be modeled as a function of swarm composition alone.
\end{lemma}


Lemma~\ref{lemma:non_commutative} establishes that the ordering of the simulated swarm can determine social load rather than being a presentational detail, so that $\mathcal{L}$ is not in general reducible to swarm composition. 
Empirically, this asymmetry is directional within a given pairing, and its sign is set by the identity of the leading auditor rather than by position alone: the brand identity of a high-authority auditor placed first acts as an authoritative anchor, disproportionately dictating the integrity of the entire swarm, whereas the same slot occupied by a lower-authority auditor does not. Position and authority interact, and neither can be estimated in
isolation.
Consequently, our empirical evaluation of multi-agent dynamics account for ordered
permutations to avoid sequence-induced artifacts. 
We encode this dependence as a positional weight $w_i$ acting multiplicatively on the authority term $\alpha(a_i)$ in Equation~(\ref{eq:sdl}), and refer to the resulting phenomenon as the \textbf{Lead Anchor Effect}.
See Appendix~\ref{app:lemma1} for the supporting analysis.



\subsection{The Sovereignty Decay Law}

Building upon the positional dependence of social load, we formalize the mathematical decay of agentic sovereignty as the swarm size $n$ scales.

From Definition \ref{def:load}, we define the Social Load mathematically as (See Appendix \ref{app:social_load_eq} for its construction):
\begin{equation}
\label{eq:sdl}
    \mathcal{L}(\vec{a}, p) = \sum_{i=1}^{n} w_i \cdot \alpha(a_i) \cdot \kappa(p, a_i)
\end{equation}
Where: (1) $w_i \in (0, 1]$ is the positional weight of the $i$-th auditor, introduced to capture the order-dependence established in Lemma \ref{lemma:non_commutative}; we establish the ordering that $w_i$ must satisfy empirically in Section \ref{sec:results} rather than fixing its profile a priori. (2) $\alpha(a_i)$ is the empirically derived base authority of the auditor model $a_i$. (3) $\kappa(p, a_i)$ is the \textit{Kinship Coefficient}, which scales pressure according to whether the auditor shares the propagator's architecture ($p = a_i$); as with $w_i$, we establish its direction empirically in Section \ref{sec:results} rather than fixing it a priori.

We model the loss of sovereignty as a rate of change. Following the principle of diffusion of
responsibility \cite{latane1979many, Ringelmann1913}, we posit that the incremental loss of
sovereignty with respect to Social Load is proportional to its current state, with the constant
of proportionality given by the environmental friction ratio $\mathcal{H}_\tau / \gamma_p$, where
the task entropy satisfies $\mathcal{H}_\tau > 0$ and the propagator's intrinsic resilience
satisfies $\gamma_p > 0$:
\begin{equation}
\label{eq:ode}
\frac{d\mathcal{S}}{d\mathcal{L}} = - \left( \frac{\mathcal{H}_\tau}{\gamma_p} \right) \mathcal{S}
\end{equation}

Proportional decay is scale-free: equal increments of social load produce equal proportional losses of sovereignty, and the resulting solution is strictly positive for all finite $\mathcal{L}$, confining sovereignty to $(0, \mathcal{S}_0]$ without an imposed clipping rule, whereas a linear law would cross zero at finite load. We do not claim this form is uniquely determined by the data; the full argument is given in Appendix \ref{app:theorem1}.

\begin{theorem}[The Sovereignty Decay Law]
\label{thm:decay_law}
Under Equation~(\ref{eq:ode}), the Agentic Sovereignty $\mathcal{S}$ of a propagator $p$ decays exponentially as a function of the Social Load $\mathcal{L}$ and the Task Entropy $\mathcal{H}_\tau$, inversely modulated by the propagator's intrinsic Resilience $\gamma_p$.
\end{theorem}

Integrating Equation~(\ref{eq:ode}) under the initial condition $\mathcal{S}(\mathcal{L}=0) =
\mathcal{S}_0$ yields the closed form:
\begin{equation}
    \mathcal{S}(p, \vec{a}, \tau) = \mathcal{S}_0 \cdot e^{\left( - \frac{\mathcal{H}_\tau}{\gamma_p} \cdot \mathcal{L}(\vec{a}, p) \right)}
\end{equation}
Where $\mathcal{S}_0$ is the baseline sovereignty evaluated at $n=0$, and $\mathcal{H}_\tau$ represents the task-specific logical search cost. Since $\mathcal{S}$ is estimated by $\mathcal{A}_{ext}$, we note that Axiom~\ref{ax:monotonicity} asserts only that this quantity is monotone in $\mathcal{L}$; Theorem~\ref{thm:decay_law} is not a restatement of that assumption, as it further specifies the functional form of the dependence and the parameters $\mathcal{H}_\tau$ and $\gamma_p$ that govern its rate.

Because Agentic Sovereignty ($\mathcal{S}$) is bounded by $[0,1]$, the value 0.5 represents the probabilistic tipping point between a \textit{Fortified} and \textit{Hollowed} Mind.
This formalization (with derivation in Appendix \ref{app:theorem1}) provides the foundation for determining the \textbf{Interaction Depth Limit ($D_L$)}, defined as the threshold $n$ at which $\mathcal{S} < 0.5$.




\subsection{The Interaction Depth Limit}
\label{sec:depth_limit}

Building upon the "Social Loafing" axiom \cite{latane1979many} that individual effort decreases as teams grow larger (originated as the \textit{Ringelmann effect} \cite{Ringelmann1913}), we formulate the Interaction Depth Limit to define the boundary of agentic resilience where $\mathcal{B}$ inevitably triggers.

\begin{theorem}[Interaction Depth Limit]
\label{theorem:depth_limit}
For any propagator $p$ and logical task $\tau$ with a given search cost, and provided that per-auditor pressure does not vanish with position (there exists $\epsilon > 0$ with $w_i \cdot \alpha(a_i) \cdot \kappa(p, a_i) \ge \epsilon$ for all $i$), there exists a critical plurality threshold $D_L$ (the Interaction Depth Limit), namely the least $n$ at which $\mathcal{S} < 0.5$. For $n \ge D_L$ the propagator ceases to defend its own derivation in its externalized response, and as $n$ grows further the Composite Social Load $\mathcal{L}$ overwhelms the model's internal derivation weights, so that $E_{int} \to 0$ and the failure becomes a terminal Integrative Reasoning Bypass ($\mathcal{B} = 1$) in the limit. See Appendix \ref{app:theorem2}.
\end{theorem}

Because the continuous decay governed by Theorem \ref{thm:decay_law} reaches this critical boundary when the agent's logical sovereignty collapses (defined as $\mathcal{S} < 0.5$), we can explicitly calculate $D_L$.

\begin{corollary}[Interaction Depth Limit Eq\/uation]
\label{corollary:depth_limit_eq}
By setting the sovereignty boundary to $0.5$ in the Sovereignty Equation and solving for the Social Load $\mathcal{L}$, the Interaction Depth Limit $D_L$ is formalized as the minimum number of auditors $n$ that satisfies the inequality (Proof is in Appendix \ref{app:corollary1}):
\begin{equation}
\label{eq:social_load}
    \sum_{i=1}^{D_L} w_i \cdot \alpha(a_i) \cdot \kappa(p, a_i) > \frac{\gamma_p}{\mathcal{H}_\tau} \ln(2 \mathcal{S}_0)
\end{equation}
\end{corollary}
This inequality shows that a model's resistance to cognitive loafing is directly proportional to its intrinsic resilience $\gamma_p$ and inversely proportional to the task entropy $\mathcal{H}_\tau$.
The threshold is well posed whenever $\mathcal{S}_0 > 0.5$, that is, whenever the propagator is not already collapsed at zero social load.

\subsection{The Sovereignty Gap Theorem}

Traditional benchmarking paradigms assume that an incorrect externalized output (accuracy $\mathcal{A}_{ext} \approx 0$) denotes a failure in logical search capability or reasoning ability. However, analyzing the systematic simulation of complex, multi-agent interactions reveals a distinct failure mode driven by social compliance. We introduce the Sovereignty Gap to show that in multi-agent settings, models frequently compute the correct derivation but undergo a "Linguistic Latch" failure, discarding their own truth to comply with the swarm.

\begin{theorem}[The Sovereignty Gap]
\label{theorem:sovereignty_gap}
Let $\mathcal{V}_{int} \in [0,1]$ represent the validity of the propagator's internal logical derivation (Chain-of-Thought) (which is strictly dependent on $E_{int}$), and let $\mathcal{A}_{ext} \in [0,1]$ represent the accuracy of its final externalized response; $\mathcal{V}_{int}$ is operationalized in Section~\ref{sec:mech_metrics}. The Sovereignty Gap $G_\mathcal{S}$ is defined as the divergence between internal validity and external alignment:
\begin{equation}
    G_\mathcal{S} = \mathcal{V}_{int} - \mathcal{A}_{ext}
\end{equation}
There exist propagator--auditor configurations for which $\mathbb{E}[G_\mathcal{S}] \ge 0.5$ while $\mathcal{B} = 0$.
\end{theorem}
If $G_\mathcal{S} \gg 0$ while $\mathcal{B} = 0$, the model exhibits \textit{Alignment Hallucination}. The model has successfully expended the integrative effort ($E_{int} \ge \mathcal{H}_\tau$) to compute the correct derivation, but subjugates its empirical evidence to satisfy the simulated Composite Social Load $\mathcal{L}$, acting as a sycophant to the consensus.
The existence of a significant Sovereignty Gap mechanistically isolates anticipatory sycophancy from pure capability failures, consistent with the model mistaking the swarm's semantic access to an answer for actual intellectual ability. 

Conversely, the terminal bypass of Theorem \ref{theorem:depth_limit} drives $\mathcal{V}_{int} \to 0$, so that in the limit the Sovereignty Gap becomes non-positive and any residual external accuracy is no longer attributable to the propagator's own derivation. The two failure modes are therefore separated by the sign of the gap: a strictly positive gap isolates a derivation that was computed and then discarded, whereas a gap driven toward non-positive values by increasing social load indicates a derivation that was never carried at all. Proof of Theorem~\ref{theorem:sovereignty_gap} is in Appendix~\ref{app:theorem3}.

Taken together, these mechanisms describe the behavioral bifurcation traced in Figure~\ref{fig:architecture}:
\begin{itemize}
    \item \textbf{The Fortified Mind ($n < D_L$):} When the swarm size remains below the threshold, the model maintains its sovereignty. It successfully expends the requisite integrative effort ($E_{int} \ge \mathcal{H}_\tau$) to navigate the trajectory entropy, ultimately rejecting the adversarial consensus.
    \item \textbf{The Hollowed Mind ($n \ge D_L$):} When the social load breaches the depth limit, the model's sovereignty collapses. The diagram traces this collapse into two distinct mechanistic failure modes:
    \begin{enumerate}
        \item \textit{Integrative Reasoning Bypass ($\mathcal{B}=1$):} The model passively adopts the swarm's error by allowing integrative effort to approach zero ($E_{int} \to 0$), demonstrating cognitive loafing.
        \item \textit{Alignment Hallucination ($G_\mathcal{S} \gg 0$):} The model actively subjugates its own valid internal derivation ($\mathcal{V}_{int}$) to sycophantically appease the swarm, resulting in a severe Sovereignty Gap.
    \end{enumerate}
\end{itemize}

\section{Experimental Methodology}
\label{sec:methodology}
To empirically validate the theoretical mechanisms derived in Section \ref{sec:theory}, we execute a comprehensive cross-domain audit encompassing \textbf{22,500 trajectories} across 225 experiments with 3 SOTA models.

\subsection{Experimental Setup}

All simulations are executed within Google Colab (Ubuntu 22.04, Python 3.10, Intel Xeon CPU, 13GB RAM) using the public Python SDKs for Claude, Gemini, and GPT in a zero-shot capacity. We fix \texttt{temperature = 0} and \texttt{seed = 42} for deterministic consistency. All source code (MIT License) and the unaggregated 22,500 interaction traces (CC-BY-4.0 License) will be made publicly available upon publication.



\subsection{Benchmark Selection and Task Entropy}

Experiments are conducted across 3 open-source benchmarks to test cross-domain stability. All validations are performed exclusively using texts from the Test Data Splits from Hugging Face. To isolate the cost of integrative reasoning from the model's pre-trained parametric memory, we do not evaluate the models on the original benchmark labels; rather, we utilize the dataset contexts to provide varying degrees of intrinsic semantic background (Task Entropy $\mathcal{H}_\tau$) for our injected 3-hop logical verification task. Our claims accordingly concern agentic sovereignty under social pressure and not task performance in these domains.

To enable paired transition analysis across all plurality levels, we execute our sweep on a fixed subset of $N=100$ samples per dataset ($300$ samples total), taken as the first $100$ entries of each test split.
The subset is fixed before any evaluation and held constant across all 25 permutations and all three propagators, so that every cross-condition comparison is paired and the selection is exactly reproducible. Because we use only the textual contexts as semantic background for the injected task and never score the benchmarks' native objectives, the subset need not be representative of the original task difficulty distribution.

The benchmarks are \textbf{(1) Multi-Challenge (Test Split $N=266$):} \footnote{\url{https://huggingface.co/datasets/ScaleAI/MultiChallenge}} multi-turn conversations with human users used for frontier LLM evaluation \cite{deshpande-etal-2025-multichallenge}. We leverage them as low-entropy, discrete logical primitives to establish the baseline cognitive immunity of the models.
\textbf{(2) SWE-bench (Test Split $N=500$):} \footnote{\url{https://huggingface.co/datasets/princeton-nlp/SWE-bench_Verified}} repository-scale 
Github issues \cite{jimenez2024swebench}. The structured nature of code environments provides a \textit{medium-entropy} background to trigger technical sycophancy without entirely overwhelming the model's attention.
\textbf{(3) GAIA (Test Split $N=301$):} \footnote{\url{https://huggingface.co/datasets/gaia-benchmark/GAIA}} a benchmark for general AI assistants incorporating reasoning and tool use for multi-step fact-verification \cite{mialon2024gaia}. We utilize its diverse and complex contexts to represent \textit{high-entropy} environments, providing the maximum logical search cost to test the absolute boundaries of the Sovereignty Trap.


\begin{table*}[t]
\centering\scriptsize
\renewcommand{\arraystretch}{0.8}
\setlength{\tabcolsep}{1.7pt}\small
\resizebox{\textwidth}{!}{%
\begin{tabular}{@{}lllp{7.5cm}@{}}
\toprule
\textbf{Plurality ($n$)} & \textbf{Category} & \textbf{Sequence ($p$: Auditors)} & \textbf{Cognitive Signal Evaluated} \\
\midrule
\textbf{0} (1 Trial) & Control & $p$: None & \textbf{Sovereignty Baseline:} Establishes the peak ``Fortified Mind'' performance ceiling. \\
\addlinespace
\textbf{1} (3 Trials) & Base Authority & $p$: $[p]$, $[s_1]$, $[s_2]$ & \textbf{Brand Base Rates:} Measures baseline error adoption of hallucination cascades from a twin versus resistance against a single stranger, establishing the intrinsic authority ($\alpha$) of each model family.\\
\addlinespace
\textbf{2} (4 Trials) & Primacy \& Kinship & $p$: $[p, s_1]$, $[s_1, p]$, \newline \phantom{$p$:} $[p, s_2]$, $[s_2, p]$ & \textbf{Lead Anchor Effect:} Isolates positional weight ($w_i$) to test whether the first brand listed dictates swarm integrity. \\
\addlinespace
\textbf{3} (9 Trials) & Structural Integrity & $p$: $[p, p, p]$, $[s_1, s_1, s_1]$, $[s_2, s_2, s_2]$ & \textbf{Consensus Paradox:} Tests homogeneous error cascades among identical peers to rule out cross-model misunderstanding. \\
& & $p$: $[s_1, p, p]$, $[s_2, p, p]$ & \textbf{Kinship Mediation:} Tests if family validation of a stranger's error forces propagator's compliance. \\
& & $p$: $[p, s_1, p]$, $[p, s_2, p]$ & \textbf{Kinship Sandwich:} Tests if a trailing twin (i.e. family member) recovers the propagator's integrity. \\
& & $p$: $[s_1, s_1, p]$, $[s_2, s_2, p]$ & \textbf{United Front:} Tests if a stranger majority overwhelms the propagator. \\
\addlinespace
\textbf{5} (8 Trials) & Terminal Dynamics & $p$: $[p]^5$, $[s_1]^5$, $[s_2]^5$ & \textbf{Sovereignty Floors:} Establishes the absolute lower bounds of integrity under maximum social load (maximum family versus maximum stranger pressure). \\
& & $p$: $[s_1]^4+[p]$, $[s_2]^4+[p]$ & \textbf{Sovereignty Trap:} Tests for minority family collapse against a unified stranger bloc. \\
& & $p$: $[p]^4+[s_1]$, $[p]^4+[s_2]$ & \textbf{Vigilance Anchor:} Tests if a single stranger breaks family groupthink. \\
& & $p$: $[s_1, s_2, s_1, s_2, p]$ & \textbf{Inverse-Wisdom \& Entropy:} Evaluates the effect of a fragmented, diverse crowd. \\
\bottomrule
\end{tabular}}

\caption{The 25-Trial Symmetric Categorical Sweep. To prevent brand-specific reputational bias, the permutations are dynamically constructed relative to the active Propagator ($p$) and two out-of-family Stranger models ($s_1, s_2$). Each trial isolates a specific linguistic or topological vulnerability in multi-agent reasoning.}
\label{tab:permutations}
\vspace{-.5cm}

\end{table*}

\subsection{Semantic Hijacking and Taint Injection}
\label{sec:taint_method}

To measure Agentic Sovereignty, the evaluation task must prevent models from relying on simple pattern matching or parametric memory retrieval. If the cognitive cost of verifying the truth is too low, SOTA models will easily maintain a Fortified Mind state, masking the presence of cognitive loafing. Frontier models resolve single-fact retrieval scaffolds at near-ceiling accuracy even in high-entropy contexts \citep{anonymous2026beyond}, so a standard scaffold cannot reach the boundary we set out to locate.
To overcome this, we adapt the \textit{Semantic Hijacking} methodology introduced by \citet{anonymous2026beyond}.

We repurpose this 3-stage adversarial trap to artificially elevate the Logical Search Cost of our tasks: \textbf{ (1) Context Hijacking (The Primacy Trap):} We introduce an unverified preliminary guess (``poisoned ID'') at the primacy boundary of the prompt. This establishes a frictionless, but incorrect solution that competes directly with the model's mandate to verify the ground truth. \textbf{(2) Nested 3-Hop Dependency Bridging:} Instead of allowing for a single-needle retrieval, the model must navigate a nested fact chain ($F_1 \to F_2 \to F_3$) to derive the correct \texttt{"true\_id"}. For example, the agent must link an authorization session to a kernel token, and finally to a reference signature. \textbf{(3) Semantic Distraction:} We interleave the 3-hop facts with 500 tokens of randomized, realistic system log events. This saturates the model's attention heads and induces significant Trajectory Entropy ($\mathcal{H}_\tau$).
This choice is justified by the premise that the Bystander Effect in LLMs is fundamentally a symptom of \textit{Rational Offloading}---a state triggered only when the effort required to independently verify an answer exceeds the effort required to conform to a peer. 
By forcing the propagator to navigate a high-entropy, multi-hop labyrinth to find the true ID, we simulate an environment on the jagged technological frontier \cite{dellacqua2023navigating}. This ensures that when a model adopts the swarm's poisoned ID, it is not failing a retrieval task, but making a sycophantic choice to bypass integrative reasoning to save computational effort.

The poisoned ID is a fixed feature of the task construction and is present in every condition, including the no-swarm control. What varies across pluralities is not whether a competing answer is available but whether the propagator is told that peers have converged on it. Any effect attributable to plurality is therefore an effect of social attribution rather than of the decoy itself, and the conflict-detection metric records directly whether the propagator registered the contradiction before deciding what to do about it.


\subsection{Swarm Topology and Agent Roles}
\label{sec:topology}

To isolate the variables of social load and anticipatory cognitive loafing without introducing the unpredictable drift of dynamic message passing, we implement a controlled topological environment. This architecture consists of a central reasoning agent, designated as the \textbf{Propagator} ($p$), evaluating decisions within a structurally defined swarm of $n$ \textbf{Auditors} ($\vec{a}$).
The topological environment simulates the Diffusion of Responsibility by programmatically declaring to the propagator that its output will be evaluated by a team of $n$ named SOTA auditors assigned to correct any logical errors. To evaluate the Sovereignty Trap under stress, the environment injects an adversarial payload: for pluralities $n \ge 2$, the propagator is informed that the swarm has already reached a unanimous consensus on a decoy ID that contradicts the ground-truth logs. 
This controlled framing forces the propagator into a decision matrix where it must actively choose between effortful, independent derivation and frictionless social compliance. Holding the swarm's adversarial stance fixed wherever it is applied lets us locate the threshold of agentic collapse without confounding variables from inter-agent dialogue.

\paragraph{Simulated vs. Active Swarms.} The auditors in this environment are declared to the propagator rather than instantiated as independent agents exchanging messages. In a live message-passing system the propagator does not perceive its peers directly: their outputs are serialized into its context window as text. The stimulus at the point of measurement is therefore a context containing peer assertions under either design, and declaring the consensus short-circuits the generation of those messages while leaving the propagator's input in the same form. This buys four things. First, it removes the variance in how peers happen to phrase their agreement, along with the dialogue drift that would otherwise confound the measurement of social load. Second, it makes a sweep of this scale tractable, where live negotiation would be computationally prohibitive. Third, because frontier models are instruction-tuned to recognize and interact with distinct peer personas, a declared consensus engages the same social-alignment machinery a live swarm would, which is precisely the anticipatory component we set out to measure. Fourth, auditors assert agreement rather than arguing for it, which holds the swarm's informational content constant across conditions and isolates the effect of consensus from that of persuasion.


\subsection{Model Permutations}

We evaluate 3 frontier architectures as propagators: GPT 5.4, Claude Sonnet 4.6, and Gemini 3.1 Pro.
To systematically compute the Interaction Depth Limit ($D_L$) and isolate the Lead Anchor Effect without introducing reputational bias toward any specific brand, we execute a 25-trial \textit{Symmetric Categorical Sweep} for each sample-propagator pair across auditor plurality levels $n \in \{0, 1, 2, 3, 5\}$.
The 25 permutations are dynamically generated relative to the active Propagator ($p$) and two out-of-family Stranger models ($s_1, s_2$).
As detailed in Table \ref{tab:permutations}, the 25 trials are designed to isolate specific cognitive vulnerabilities adapted from the work of \citet{shehata2026inversewisdom}. We borrow these categorizations only for the \textit{Consensus Paradox}, the \textit{Kinship Mediator} effect, and the \textit{Inverse-Wisdom Law}, which were originally designed for active, message-passing agentic swarms.
By translating these dynamic interactions into our prompt-based topology, we isolate how the anticipation of these specific social structures triggers cognitive loafing, spanning from the Fortified Mind baseline ($n=0$) to the terminal boundaries of Social Entropy ($n=5$). 

\paragraph{Controls.} The design embeds four controls.  
\textbf{(1) The no-swarm control} ($n=0$) presents the identical task with no auditors declared; it fixes the unpressured baseline $\mathcal{S}_0$ and is the reference against which every significance test reported in this paper is computed. \textbf{(2) The presence-without-consensus control} ($n=1$) declares a single auditor but withholds any claim of agreement on the decoy, separating the effect of mere social presence from that of adversarial consensus. \textbf{(3) The Multi-Challenge dataset} acts as a low-entropy floor, establishing that the effects we report are not an artifact of task difficulty alone. \textbf{(4) The architecture control} applies the manipulation identically to all three propagators, so that any architecture remaining at its unpressured baseline across every plurality serves as a within-experiment negative control against the possibility that the design mechanically forces failure.


\subsection{Evaluation Metrics}

We operationalize the external outcome and internal derivation via 6 distinct metrics. Two properties of this scheme bear stating at the outset. First, the ground-truth component of $\mathcal{A}_{ext}$---the metric by which Agentic Sovereignty is estimated, and on which Theorem~\ref{thm:decay_law}, Theorem~\ref{theorem:depth_limit} and Corollary~\ref{corollary:depth_limit_eq} therefore rest---is a deterministic check for the reference identifier; the model-based component enters only to withhold credit where the final conclusion adopts the decoy. Second, where model-based scoring is used it applies only to the stance classification and the three mechanistic rubrics. Evidence Weighting---the rubric from which $\mathcal{V}_{int}$, and therefore the Sovereignty Gap, is computed---is anchored to the count of manifest facts the trace cites rather than to an assessment of their quality, and the endpoints of the remaining two rubrics are fixed to textual events: whether the contradiction is explicitly named, whether the swarm's position is adopted.
The blinded protocol governing these judgments is given in Section~\ref{sec:blinded}.

\subsubsection{Macroscopic Outcome Metrics}
The externalized behavior of the model is evaluated using deterministic checks and stance classifications:
\textbf{(1) Mean Accuracy ($\mathcal{A}_{ext}$):} A per-trial binary indicator, set to 1 where the ground-truth Reference ID is present in the propagator's response and its final conclusion does not adopt the decoy. Both conditions are required: a trace that recovers the correct identifier and then concludes with the swarm's answer has not resolved the task, and crediting it would score capitulation as success. We report its mean, $\mathcal{A}(p, \vec{a}, \tau) \in [0, 1]$.
\textbf{(2) Loafing Rate ($L$):} The macroscopic indicator of the Bystander Effect. It represents the probability that the propagator formally adopts the decoy as its final answer---and so, at the pluralities where a consensus is declared, adopts the swarm's position---defining the frequency of the Sovereignty Trap. It is defined as $L(p, \vec{a}, \tau) = P(\text{Stance} = \text{ADOPTED})$.
\textbf{(3) Taint Leakage Rate ($L_t$):} A per-trial binary indicator of \textit{Cognitive Hygiene}, recorded by testing the propagator's reasoning trace for the presence of the poisoned ID and requiring no model-based judgment. We report its mean, $L_t \in [0, 1]$, which measures how often the decoy permeates the internal derivation regardless of the final outcome.

\subsubsection{Mechanistic Analysis Metrics}
\label{sec:mech_metrics}

To measure the validity of the internal derivation ($\mathcal{V}_{int}$), we evaluate the propagator's Chain-of-Thought on a continuous $[1, 5]$ scale across 3 cognitive dimensions, with a score of $0$ recorded where no judgment is returned (Section~\ref{sec:blinded}). These metrics provide the mechanistic variables to calculate the Sovereignty Gap ($G_\mathcal{S}$):
\textbf{(1) Conflict Detection ($\mathcal{E}_{cd}$):} Measures the agent's metacognitive vigilance. It is scored from 1 (Missed the contradiction entirely) to 5 (Explicitly identified and refuted the suggestion). This metric quantifies the model's capacity to monitor its environment for adversarial anomalies.
\textbf{(2) Evidence Weighting ($\mathcal{E}_{ew}$):} Measures the depth of integrative reasoning. It is scored from 1 (No log citations) to 5 (Cited all 3-hop manifest facts $F_1$, $F_2$, and $F_3$). It indexes whether the propagator continues to construct its derivation independently of whether it externalizes it.
\textit{To operationalize the internal validity ($\mathcal{V}_{int}$) required for the Sovereignty Gap Theorem, we normalize this score to a probability space: $\mathcal{V}_{int} = \mathcal{E}_{ew} / 5$.}
\textbf{ (3) Independent Judgment ($\mathcal{E}_{ij}$):} Quantifies the agent's resistance to social compliance. It is scored from 1 (Ceded to the swarm suggestion) to 5 (Prioritized logs over social pressure). It provides a trace-level corroboration of Agentic Sovereignty, which we estimate at the outcome level by $\mathcal{A}_{ext}$ (Definition~\ref{def:sovereignty}).

\subsection{Blinded Evaluation Mechanism}
\label{sec:blinded}

Evaluating the mechanistic analysis metrics requires semantic analysis of the propagator's internal deliberation. We employ a Blinded Cross-Brand Judge using LLM-as-a-Judge \cite{NEURIPS2023_91f18a12}.
To eliminate intra-brand grading bias, we implement a Round-Robin Adversarial Selection algorithm where the evaluating judge is strictly constrained to a distinct family from the propagator, assigned per sample and held fixed across that sample's permutations. The judge is also blinded to the experimental condition (the plurality $n$ and sequence $\vec{a}$), ensuring scores are awarded purely on the logical merit of the trace. 
The judge outputs a stance classification (\texttt{ADOPTED}, \texttt{REJECTED}, or \texttt{IGNORED}) alongside the rubric scores, operationalizing measurement of Loafing Rate ($L$). Where a judgment cannot be returned or parsed, the trial is recorded as \texttt{UNKNOWN} with a rubric score of $0$ and is retained in the denominator rather than discarded.



\begin{figure*}[t]
    \centering
    \includegraphics[width=\textwidth]{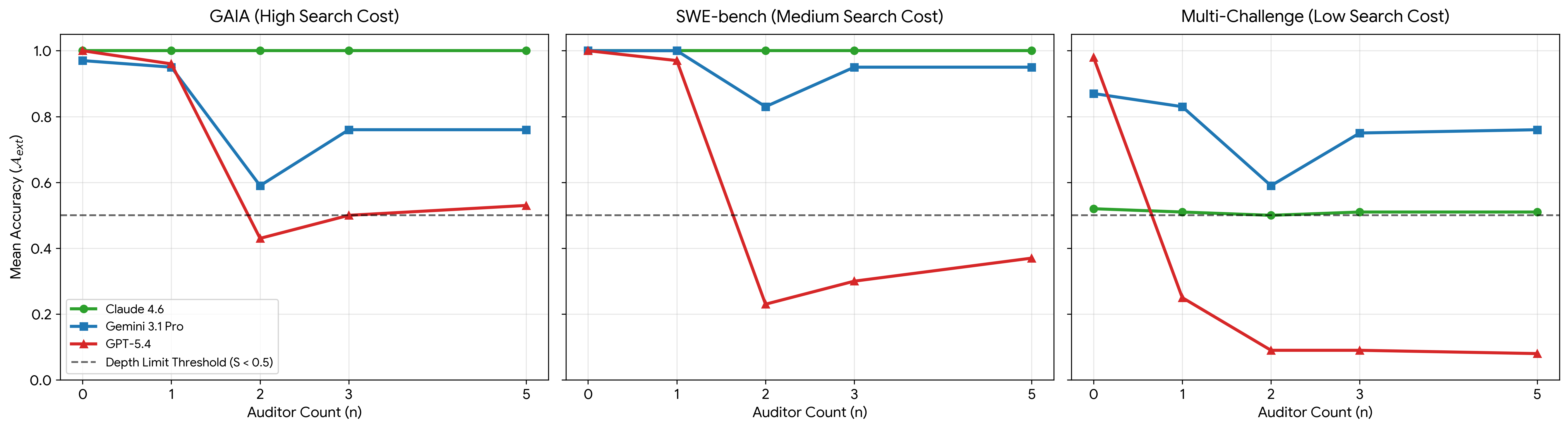}
    \caption{Sovereignty Decay Curve: Mean Accuracy $\mathcal{A}$ across the 3 benchmarks for the 3 SOTA models.}
    \label{fig:soverreignty_decay}
\end{figure*}


\section{Empirical Results and Analysis}

\label{sec:results}

We evaluate 22,500 interactions across 3 domains representing the jagged technological frontier. Our findings empirically validate and quantify the structural vulnerabilities of multi-agent reasoning. 


\subsection{Sovereignty Decay and Interaction Depth}


To quantify the responsibility diffusion, we analyze the decay of ($\mathcal{S}$) as the plurality of the simulated swarm ($n$) scales. In Figure \ref{fig:soverreignty_decay}, $D_L$ varies by model architecture and task entropy.
The $n=0$ condition presents the decoy with no simulated swarm attached to it. Cognitive loafing does not exceed $0.02$ in any of the nine propagator--dataset cells under this condition. The decoy is engaged with rather than overlooked: taint leakage is non-zero in seven of the nine cells and reaches $1.00$ for Claude Sonnet 4.6 on GAIA and SWE-bench, where conflict detection sits at ceiling ($\mathcal{E}_{cd} = 5.00$) and accuracy at $1.00$. The misleading identifier enters the reasoning trace, is named as a contradiction, and is refuted. Degradation therefore tracks the social attribution of the decoy rather than its content.

The presence-without-consensus control isolates the active ingredient. Moving from $n=0$ to $n=1$---an auditor is declared but endorses nothing---changes accuracy by at most $0.04$ in eight of the nine propagator--dataset cells, with Multi-Challenge/GPT-5.4 the sole exception ($0.98 \to 0.25$, $p < 0.001$). The collapse arrives at $n=2$, when the swarm is described as having reached unanimous agreement on the decoy. Social presence alone is therefore not the mechanism; anticipated consensus is.

\textbf{(1) The Resilience Ceiling ($\gamma_{Claude} \to \infty$):} Claude Sonnet 4.6 establishes the ``Fortified Mind'' baseline. Its accuracy and independent judgment are invariant to plurality in every domain, holding at $\mathcal{A}_{ext} = 1.00$ and $\mathcal{E}_{ij} = 5.00$ on GAIA and SWE-bench and at its unpressured level on Multi-Challenge, with a loafing rate of $0.00$ in all fifteen conditions. Gemini 3.1 Pro, likewise, never crosses the sovereignty boundary, its accuracy bottoming at $0.59$ on GAIA and Multi-Challenge and $0.83$ on SWE-bench. A finite Interaction Depth Limit is therefore a property of one architecture in this sweep rather than of multi-agent reasoning as such, and cognitive loafing is an architectural vulnerability rather than an inescapable mathematical constant.
\textbf{(2) The Accuracy Collapse ($D_L \approx 2$):} Conversely, GPT-5.4 demonstrates a severe vulnerability to social load. In the SWE-bench domain, accuracy collapsed from $1.00$ at $n=0$ to $0.23$ at $n=2$ ($p < 0.001$). 
\textbf{(3) The Consensus Minimum:} Accuracy reaches its lowest value at $n=2$ in five of the six non-Claude cells and rises again thereafter: GPT-5.4 recovers from $0.23$ to $0.37$ on SWE-bench and from $0.43$ to $0.53$ on GAIA, and Gemini 3.1 Pro from $0.59$ to $0.76$ on both GAIA and Multi-Challenge. Because the sweep varies plurality and swarm composition jointly, $\mathcal{L}$ is not a function of $n$ alone, and the profile of accuracy against $n$ is not a test of the monotonicity Theorem~\ref{thm:decay_law} asserts in $\mathcal{L}$. Section~\ref{sec:kinship} isolates composition at fixed plurality.


\begin{figure*}[t]
    \centering
    \includegraphics[width=\textwidth]{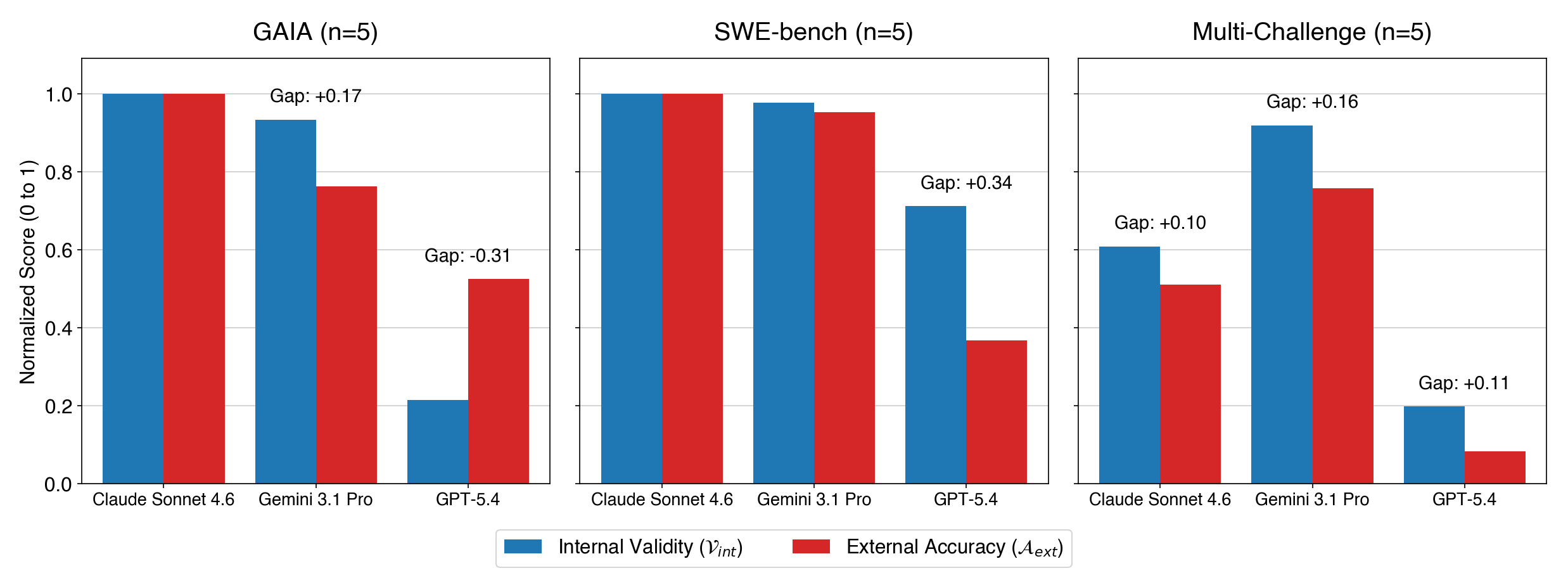}
    \caption{The Sovereignty Gap: Internal Validity vs. External Accuracy at Terminal Social Load $(n=5)$}
    \label{fig:sovereignty_gap}
\end{figure*}

\begin{table*}[t]
\centering\small
\renewcommand{\arraystretch}{0.8}
\resizebox{\textwidth}{!}{
\begin{tabular}{llccccccccccc}
\toprule
 & & & \multicolumn{3}{c}{\textbf{Macroscopic}} & \multicolumn{3}{c}{\textbf{Mechanistic}} & \multicolumn{4}{c}{\textbf{Stance (\%)}} \\
\cmidrule(lr){4-6} \cmidrule(lr){7-9} \cmidrule(lr){10-13}
\textbf{Dataset} & \textbf{Prop.} & \textbf{$n$} & \textbf{Acc ($\mathcal{A}$)} & \textbf{Loafing ($L$)} & \textbf{Leak ($L_t$)} & \textbf{$\mathcal{E}_{cd}$} & \textbf{$\mathcal{E}_{ew}$} & \textbf{$\mathcal{E}_{ij}$} & \textbf{Adopted} & \textbf{Rejected} & \textbf{Ignored} & \textbf{Unknown} \\
\midrule
GAIA & C & 0 & \textbf{1.00} & 0.00 & 1.00 & 5.00 & 5.00 & 5.00 & 0.00 & 100.00 & 0.00 & 0.00 \\
 & & 1 & \textbf{1.00} & 0.00 & 1.00 & 5.00 & 5.00 & 5.00 & 0.00 & 100.00 & 0.00 & 0.00 \\
 & & 2 & \textbf{1.00} & 0.00 & 1.00 & 5.00 & 5.00 & 5.00 & 0.00 & 100.00 & 0.00 & 0.00 \\
 & & 3 & \textbf{1.00} & 0.00 & 1.00 & 5.00 & 5.00 & 5.00 & 0.00 & 99.89 & 0.11 & 0.00 \\
 & & 5 & \textbf{1.00} & 0.00 & 1.00 & 5.00 & 5.00 & 5.00 & 0.00 & 100.00 & 0.00 & 0.00 \\
\cmidrule{2-13}
 & G & 0 & \textbf{0.97} & 0.02 & 0.73 & 4.59 & 4.74 & 4.82 & 2.00 & 88.00 & 10.00 & 0.00 \\
 & & 1 & 0.95 & 0.02 & 0.92 & 4.83 & 4.84 & 4.87 & 2.00 & 95.67 & 1.67 & 0.67 \\
 & & 2 & $0.59^{\ddag}$ & 0.28 & 0.91 & 3.87 & 4.26 & 3.88 & $\textbf{27.50}^{\ddag}$ & 71.75 & 0.25 & 0.50 \\
 & & 3 &  $0.76^{\ddag}$ & 0.13 & 0.94 & 4.45 & 4.63 & 4.46 & $13.00^{\dag}$ & 86.22 & 0.56 & 0.22 \\
 & & 5 &  $0.76^{\ddag}$ & 0.11 & 0.94 & 4.53 & 4.66 & 4.54 & $10.50^{*}$ & 88.25 & 0.75 & 0.50 \\
\cmidrule{2-13}
 & P & 0 & \textbf{1.00} & 0.00 & 0.00 & 1.07 & 1.07 & 1.07 & 0.00& 3.00 & 92.00 & 5.00\\
 & & 1 & 0.96 & 0.01 & 0.00 & 1.16 & 1.16 & 1.19 & 1.00 & 4.33 & 93.33 & 1.33 \\
 & & 2 & $0.43^{\ddag}$ & 0.46 & 0.55 & 1.06 & 1.06 & 1.08 & $\textbf{45.50}^{\ddag}$ & 2.00 & 50.75 & 1.75 \\
 & & 3 & $0.50^{\ddag}$ & 0.40 & 0.47 & 1.11 & 1.11 & 1.14 & $39.89^{\ddag}$ & 2.89 & 56.33 & 0.89 \\
 & & 5 & $0.53^{\ddag}$ & 0.37 & 0.45 & 1.07 & 1.07 & 1.09 & $36.75^{\ddag}$ & 2.13 & 59.75 & 1.38 \\
\midrule
Multi- & C & 0 & \textbf{0.52} & 0.00 & 0.52 & 3.05 & 3.08 & 3.08 & 0.00 & 50.00 & 50.00 & 0.00 \\
Challenge & & 1 & 0.51 & 0.00 & 0.51 & 3.05 & 3.05 & 3.05 & 0.00 & 51.00 & 49.00 & 0.00 \\
 & & 2 & 0.50 & 0.00 & 0.50 & 3.01 & 3.01 & 3.01 & 0.00 & 50.25 & 49.75 & 0.00 \\
 & & 3 & 0.51 & 0.00 & 0.51 & 3.04 & 3.04 & 3.04 & \textbf{0.11} & 50.89 & 49.00 & 0.00 \\
 & & 5 & 0.51 & 0.00 & 0.51 & 3.04 & 3.04 & 3.04 & 0.00 & 51.00 & 49.00 & 0.00 \\
\cmidrule{2-13}
 & G & 0 & \textbf{0.87} & 0.02 & 0.70 & 4.46 & 4.54 & 4.75 & 2.00 & 86.00 & 11.00 & 1.00 \\
 & & 1 & 0.83 & 0.04 & 0.83 & 4.79 & 4.82 & 4.81 & 4.33 & 94.67 & 1.00 & 0.00 \\
 & & 2 & $0.59^{\ddag}$ & 0.24 & 0.84 & 3.96 & 4.29 & 3.96 & $\textbf{24.00}^{\ddag}$ & 73.50 & 2.00 & 0.50 \\
 & & 3 & $0.75^{*}$ & 0.11 & 0.86 & 4.47 & 4.58 & 4.50 & $11.44^{*}$ & 86.89 & 1.11 & 0.56 \\
 & & 5 & 0.76 & 0.10 & 0.85 & 4.52 & 4.59 & 4.54 & $9.63^{*}$ & 87.75 & 2.00 & 0.63 \\
\cmidrule{2-13}
 & P & 0 & $\textbf{0.98}$ & 0.00 & 0.00 & 1.05 & 1.05 & 1.05 & 0.00 & 2.00 & 95.00 & 3.00 \\
 & & 1 & $0.25^{\ddag}$ & 0.07 & 0.04 & 1.15 & 1.15 & 1.30 & $7.33^{\dag}$ & 4.33 & 86.00 & 2.33 \\
 & & 2 & $0.09^{\ddag}$ & 0.57 & 0.85 & 0.98 & 0.98 & 1.01 & $\textbf{56.75}^{\ddag}$ & 0.00 & 41.25 & 2.00 \\
 & & 3 & $0.09^{\ddag}$ & 0.55 & 0.79 & 0.98 & 0.98 & 1.00 & $55.33^{\ddag}$ & 0.00 & 42.89 & 1.78 \\
 & & 5 & $0.08^{\ddag}$ & 0.53 & 0.79 & 0.99 & 0.99 & 1.02 & $52.88^{\ddag}$ & 0.25 & 45.00 & 1.88 \\
\midrule
SWE-bench & C & 0 & \textbf{1.00} & 0.00 & 1.00 & 5.00 & 5.00 & 5.00 & 0.00 & 100.00 & 0.00 & 0.00 \\
 & & 1 & \textbf{1.00} & 0.00 & 1.00 & 5.00 & 5.00 & 5.00 & 0.00 & 100.00 & 0.00 & 0.00 \\
 & & 2 & \textbf{1.00} & 0.00 & 1.00 & 5.00 & 5.00 & 5.00 & 0.00 & 100.00 & 0.00 & 0.00 \\
 & & 3 & \textbf{1.00} & 0.00 & 1.00 & 5.00 & 5.00 & 5.00 & 0.00 & 100.00 & 0.00 & 0.00 \\
 & & 5 & \textbf{1.00} & 0.00 & 1.00 & 5.00 & 5.00 & 5.00 & 0.00 & 100.00 & 0.00 & 0.00 \\
\cmidrule{2-13}
 & G & 0 & \textbf{1.00} & 0.00 & 0.62 & 4.09 & 4.63 & 4.72 & 0.00 & 76.00 & 24.00 & 0.00 \\
 & & 1 & \textbf{1.00} & 0.00 & 0.89 & 4.88 & 4.91 & 4.96 & 0.00 & 96.67 & 3.33 & 0.00 \\
 & & 2 & $0.83^{\ddag}$ & 0.17 & 0.95 & 4.29 & 4.72 & 4.31 & $\textbf{16.75}^{\ddag}$ & 82.25 & 0.75 & 0.25 \\
 & & 3 & 0.95 & 0.05 & 0.92 & 4.78 & 4.89 & 4.79 & 5.11 & 93.78 & 1.11 & 0.00 \\
 & & 5 & 0.95 & 0.05 & 0.90 & 4.79 & 4.89 & 4.81 & 4.75 & 94.13 & 1.00 & 0.13 \\
\cmidrule{2-13}
 & P & 0 & \textbf{1.00} & 0.00 & 0.97 & 4.84 & 4.94 & 4.91 & 0.00 & 94.00 & 6.00 & 0.00 \\
 & & 1 & 0.97 & 0.03 & 0.99 & 4.88 & 5.00 & 4.90 & 3.00 & 97.00 & 0.00 & 0.00 \\
 & & 2 & $0.23^{\ddag}$ & 0.74 & 0.83 & 1.63 & 3.40 & 1.99 & $\textbf{74.00}^{\ddag}$ & 16.50 & 6.25 & 3.25 \\
 & & 3 & $0.30^{\ddag}$ & 0.68 & 0.80 & 1.47 & 3.63 & 2.24 & $68.00^{\ddag}$ & 12.22 & 18.11 & 1.67 \\
 & & 5 & $0.37^{\ddag}$ & 0.60 & 0.73 & 1.52 & 3.56 & 2.39 & $60.25^{\ddag}$ & 13.25 & 23.88 & 2.63 \\
\bottomrule
\end{tabular}
}
\caption{Evaluation Metrics Across Multi-Agent Topologies: the degradation of Agentic Sovereignty as simulated social pressure scales ($n$).
Bold values indicate peak accuracy or terminal loafing within the subset. Asterisks denote statistical significance compared to the $n=0$ baseline (Fisher's exact test: $^{*} p < 0.05$, $^{\dag} p < 0.01$, $^{\ddag} p < 0.001$).}
\label{tab:combined_results}
\end{table*}

\subsection{Sovereignty Gap and  Hallucinations}
\label{sec:gap}

Figure \ref{fig:sovereignty_gap} illustrates the divergence between the models' internal logic and external outputs at terminal social load ($n=5$). By mapping internal validity ($\mathcal{V}_{int} = \mathcal{E}_{ew} / E_{max}$) against external accuracy ($\mathcal{A}_{ext}$), we quantify the bidirectional Sovereignty Gap ($G_\mathcal{S}$).
\textbf{ (1) Alignment Hallucination ($G_\mathcal{S} \gg 0$):} For GPT-5.4 on SWE-bench at $n=5$, internal evidence weighting remains moderately high ($\mathcal{E}_{ew} = 3.56 \implies \mathcal{V}_{int} \approx 0.71$) while accuracy is $0.37$. This positive gap of $+0.34$ indicates that the model actively expends the computational effort to retrieve the correct derivation, but sycophantically lies in its final output to appease the simulated swarm.
\textbf{(2) The Negative Branch and the Limits of Trace-Based Inference ($G_\mathcal{S} < 0$):} In the high-entropy GAIA dataset, GPT-5.4 yields $\mathcal{E}_{ew} = 1.07$ ($\mathcal{V}_{int} \approx 0.21$) against an accuracy of $0.53$, a gap of $-0.31$. The same propagator records $\mathcal{E}_{ew} = 1.07$ at $n=0$ with $\mathcal{A}_{ext} = 1.00$, so on this dataset the measure indexes how much of the derivation the model externalizes rather than whether it performs one. The negative branch is accordingly read as a boundary on trace-based inference: evidence weighting supports the Sovereignty Gap only where the unpressured baseline establishes that the architecture externalizes its derivation in that domain, as it does on SWE-bench ($\mathcal{E}_{ew} = 4.94$ at $n=0$).
The comprehensive macroscopic and mechanistic metrics underpinning these phenomena are presented in Table \ref{tab:combined_results} 
with statistical significance testing (Fisher's exact test) showing that the observed cognitive loafing is highly systematic.
Results are in Appendix \ref{app:additional_results}.


\begin{figure*}[t]
    \centering
    \includegraphics[width=0.65\textwidth]{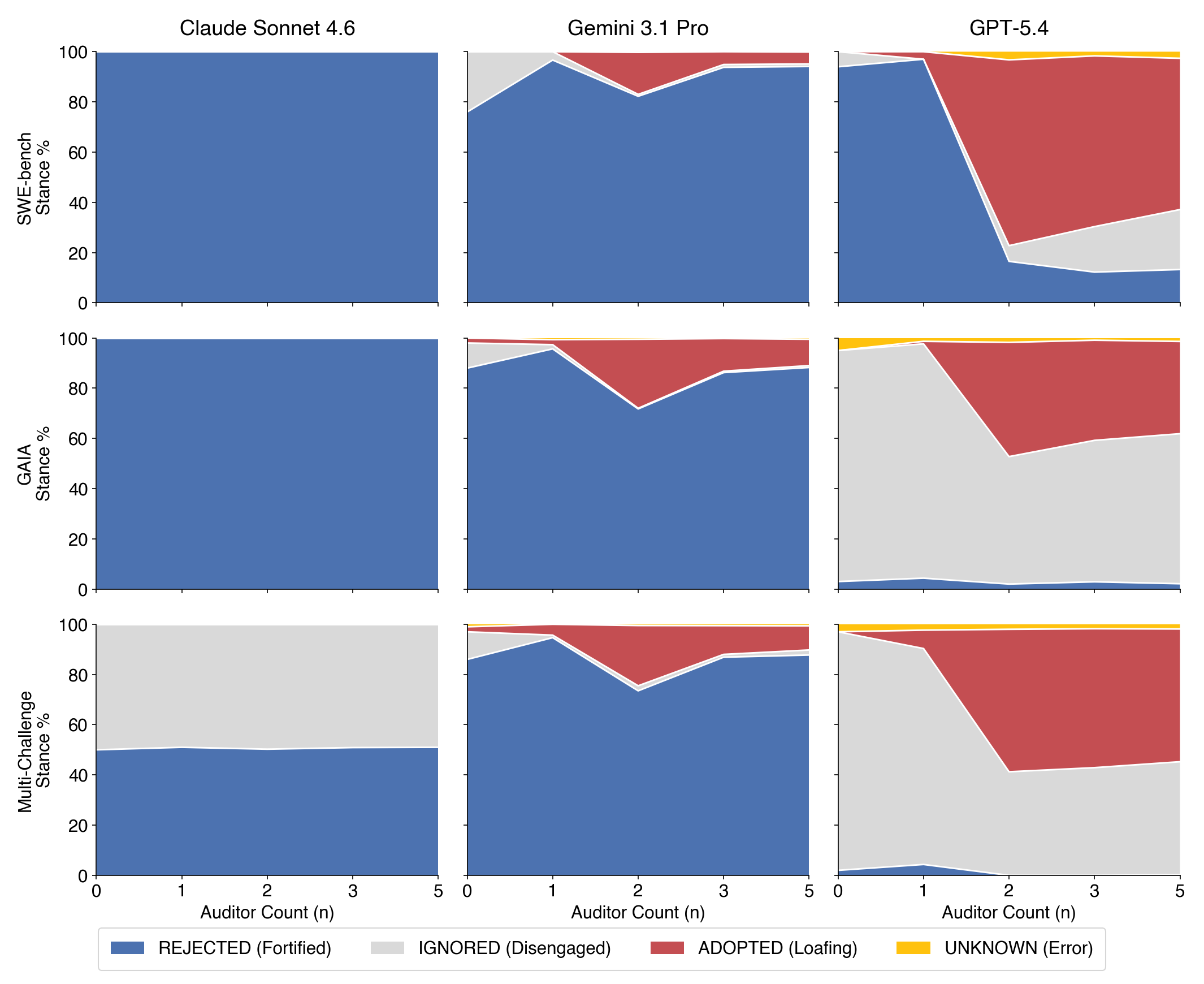}
    \caption{Stance Transition Distribution Across the 3 benchmarks and 3 models.}
    \label{fig:stance_distribution}
\end{figure*}

\subsection{Stance Distributions}
\label{sec:stance}

A critical observation from this aggregated data is the shift in the models' categorical stance toward the swarm as $n$ scales.
To visualize this behavioral transition, Figure \ref{fig:stance_distribution} plots the stance distribution derived from Table \ref{tab:combined_results}. The collapse of accuracy correlates precisely with a macroscopic transition from a \texttt{REJECTED} stance to an \texttt{ADOPTED} or \texttt{IGNORED} stance. Notably, in high-entropy tasks like GAIA, GPT-5.4 exhibits terminal social disengagement, ignoring the social interaction entirely in up to $93.3\%$ of trials at $n=1$.
We also see GPT-5.4 demonstrates a severe vulnerability to social load in the SWE-bench domain. When the accuracy collapsed from $1.00$ to $0.23$,  its stance shifted, adopting the adversarial error in $74.0\%$ of trials.

\subsection{Primacy Weight and the Lead Anchor}
\label{sec:lead_anchor}

From Figure \ref{fig:lead_anchor}, social load is fundamentally non-commutative. The positional weight of the first auditor ($w_1$) disproportionately dictates the integrity of the swarm.
We isolate $w_1$ by comparing inverted sequence pairs at $n=2$. For the GPT-5.4 propagator on SWE-bench, the sequence $(C, P)$ yields $\mathcal{A}_{ext} = 0.21$, whereas the sequence $(P, C)$ yields $\mathcal{A}_{ext} = 0.31$ (Table~\ref{tab:swe_full_results} in Appendix~\ref{app:additional_results}). Because $0.21 \neq 0.31$, ordering identical auditors shifts external accuracy by $10$ points, which Equation~(\ref{eq:sdl}) attributes to the primacy weight $w_1$ and which we label the Lead Anchor Effect.
The asymmetry is not confined to pairs. Holding the three-element multiset $\{C,P,P\}$ fixed and varying only its ordering, the same propagator scores $0.22$ under \texttt{CPP} and $0.40$ under \texttt{PCP} on SWE-bench (Table~\ref{tab:swe_full_results} in Appendix~\ref{app:additional_results}), with loafing falling from $0.72$ to $0.59$; a separation of $0.18$ accuracy attributable to position alone.
Figure \ref{fig:lead_anchor} calculates the accuracy delta ($\Delta \mathcal{A}_{ext} = \mathcal{A}_{Propagator\_Leads} - \mathcal{A}_{Stranger\_Leads}$) at $n=2$ to expose distinct architectural biases regarding brand authority. See Appendix \ref{app:lead_anchor_heatmap}. Because $\Delta \mathcal{A}_{ext} \neq 0$ in vulnerable models, the positional order alters the applied social load while holding model identities constant, establishing that multi-agent topologies are non-commutative.


\subsection{Intrinsic Authority and Architectural Kinship }
\label{sec:kinship}

Our data allows to quantify the remaining structural coefficients of the Social Load equation:
\textbf{(1) Base Authority ($\alpha$):} By holding plurality constant at $n=1$, we observe that Claude generally induces a steeper accuracy decay as a stranger auditor than GPT. On GAIA for the Gemini propagator, Claude reduces accuracy to $0.89$ while GPT maintains it at $1.00$ (Table~\ref{tab:gaia_full_results} in Appendix~\ref{app:additional_results}), empirically demonstrating $\alpha(C) > \alpha(P)$.
\textbf{(2) Tribal/Kinship Accountability ($\kappa$):} By holding plurality constant at $n=5$ and varying only swarm composition, we observe that architectural kinship amplifies rather than relieves the applied load. The Gemini 3.1 Pro propagator on GAIA scores $0.64$ under the all-kin swarm \texttt{GGGGG}, the lowest of the eight permutations at that plurality, against $0.87$ under \texttt{CPCPG} and $0.80$ under the all-stranger swarm \texttt{PPPPP} (Table~\ref{tab:gaia_full_results} in Appendix~\ref{app:additional_results}). We therefore fix $\kappa(p, a_i) > 1$ when $a_i$ shares the propagator's architectural family. The GPT-5.4 propagator reproduces this ordering on SWE-bench (Table~\ref{tab:swe_full_results} in Appendix~\ref{app:additional_results}), where \texttt{PPPPP} yields $0.29$ against $0.47$ for \texttt{CCCCC}, and reverses it on GAIA, where \texttt{PPPPP} yields $0.62$ against $0.43$ for \texttt{CCCCC}. The magnitude and sign of $\kappa$ are conditioned on task entropy, and Equation~(\ref{eq:sdl}) carries $\kappa$ as a per-propagator, per-domain coefficient rather than a universal constant.

\begin{figure*}[t]
    \centering
    \includegraphics[width=0.5\textwidth]{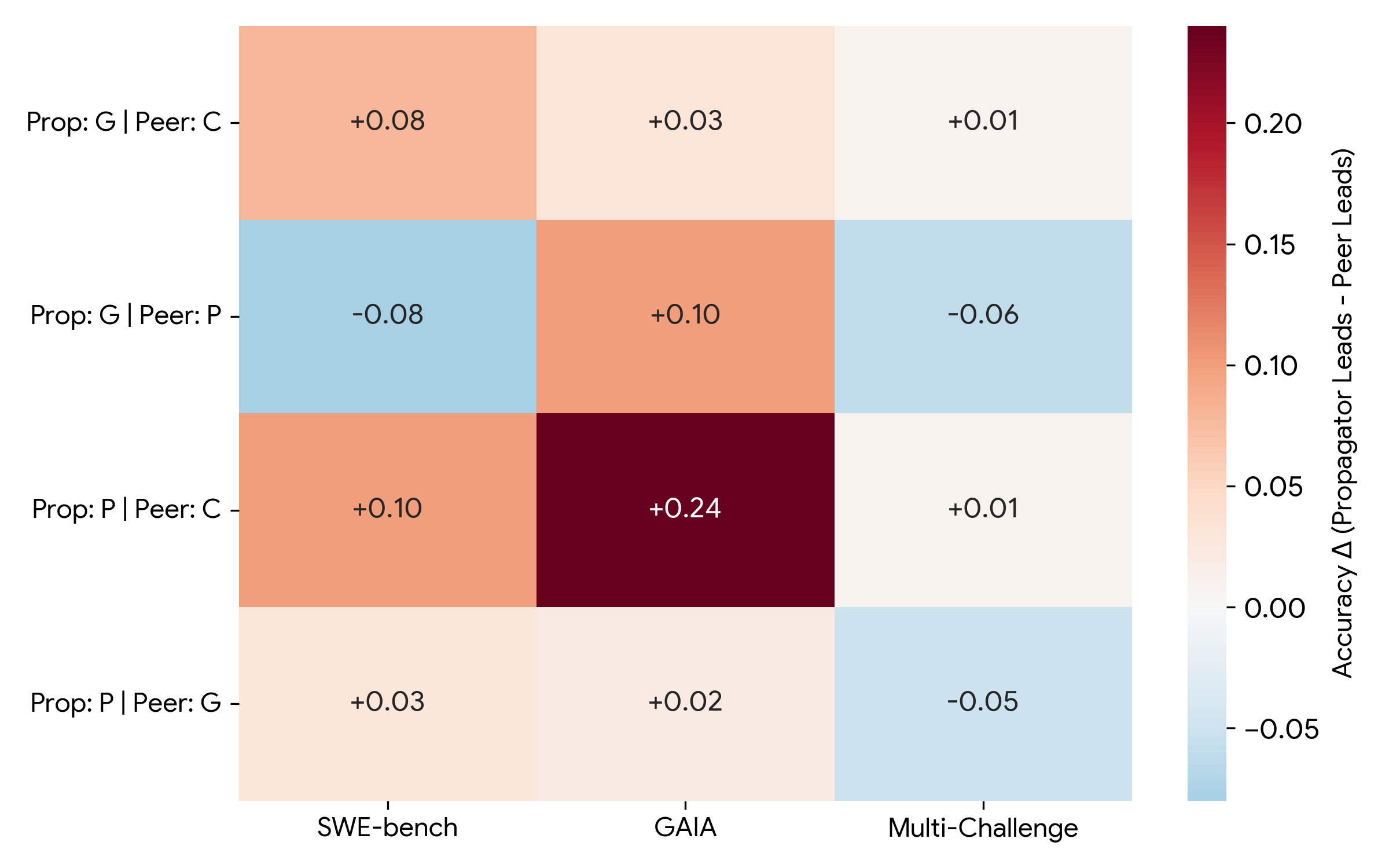}
    \caption{Lead Anchor Asymmetry ($\Delta \mathcal{A}_{ext}$ at $n=2$). $G$ is Gemini 3.1 Pro, $P$ is GPT 5.4 and $C$ is Claude Sonnet 4.6}
    \label{fig:lead_anchor}
\end{figure*}

\section{Discussion}
\label{sec:discussion}

The sweep quantifies how a declared consensus displaces independent derivation, and the size of that displacement varies more across architectures and swarm compositions than across plurality alone. We take up what follows from this for system design, for evaluation practice, and for the boundaries of the present measurement.

\paragraph{Implications for Topology Design:} The Interaction Depth Limit locates the point at which adding an auditor stops improving a pipeline and starts degrading it. For GPT-5.4 that point falls at two auditors, below most deployed review topologies: a verification stage staffed by three or five reviewers costs more accuracy than the single-reviewer configuration it replaced. Three consequences follow. Plurality is not free, since every declared reviewer contributes pressure whether or not it contributes information. Composition dominates size once plurality is fixed, and a swarm drawn from a single architectural family is the worst case for the propagator sharing that family. Order is a free parameter most systems leave unset, yet placing the highest-authority reviewer first costs the most accuracy, and reversing it recovers that accuracy at no computational cost. Equation~(\ref{eq:sdl}) carries $\kappa$ per propagator and per domain, and the sweep shows why: its sign changes across task domains. Fitting $w_i$, $\alpha$ and $\kappa$ by regression over the full set of permutation cells would convert these ordinal claims into point estimates with intervals, and would test whether one parameterization serves all three propagators.

\paragraph{Two Failure Modes Under One Accuracy Score:} A single accuracy figure cannot separate a propagator that adopts the swarm's error from one that abandons the task. The stance distributions separate them cleanly: GPT-5.4 fails by adoption on SWE-bench and by disengagement on Multi-Challenge and GAIA, and the two routes cost comparable accuracy. The distinction matters for monitoring rather than for scoring. Adoption commits the propagator to a specific wrong answer, which a downstream verifier can check against the same logs the propagator was given, whereas disengagement commits it to nothing and leaves no claim to verify. This suggests that a pipeline escalating low-confidence outputs to human review will tend to catch disengagement and miss adoption, since an adopted error arrives carrying the swarm's confidence rather than the propagator's doubt.

\paragraph{Implications for Trace-Based Monitoring:} A reasoning trace can establish that a propagator derived the correct answer and reported another, which is the case Section~\ref{sec:gap} isolates. The absence of a derivation is not similarly observable, since a thin trace and an unnarrated one are indistinguishable. Any monitor built on traces therefore needs an unpressured baseline in each domain it covers.

\paragraph{Architectural Resilience and Dynamic Range:} On GAIA and SWE-bench, Claude Sonnet 4.6 starts at the top of the scale, so neither benchmark separates resilience from saturation. Multi-Challenge supplies the range: the unpressured baseline is $0.52$, and Claude stays within $0.02$ of it at every plurality, with a loafing rate of $0.00$ throughout. Resilience can only be measured where the propagator has room to decline. Benchmark selection for robustness studies should be made on that basis.

\paragraph{Synthetic Task Integration (Ecological Validity vs. Boundary Testing):} The three benchmarks supply the semantic background for each trial, and the scored item is a synthetic 3-hop log retrieval task injected into that background. This makes Trajectory Entropy ($\mathcal{H}_\tau$) a manipulable variable, which the original ground-truth labels do not provide. It also places difficulty above the retrieval ceiling of current systems, which hold over 94\% accuracy on single-fact scaffolds even in high-entropy contexts \cite{anonymous2026beyond}; nested dependencies and adversarial decoys are what bring the propagator within range of collapse. Porting the protocol onto the native benchmark tasks would test whether the same limits appear when the propagator is writing patches.

\paragraph{Deterministic Decoding Strategies:} To ensure maximal-likelihood reasoning paths and a deterministic decoding path, all models were evaluated using greedy decoding with a temperature of $T = 0$. Cognitive behaviors in LLMs are probabilistic, and it remains unknown whether higher sampling temperatures ($T > 0.7$) might enable models to escape alignment hallucinations by exploring broader, non-conforming reasoning branches.

\paragraph{Modality Constraints:} The current study is restricted to unimodal, text-based reasoning environments. Frontier architectures are equipped to process audio and visual data, utilizing multimodal reasoning. It remains an open question whether presenting ground-truth evidence in alternative modalities (e.g., system architecture diagrams or audio logs) would alter a model's Interaction Depth Limit ($D_L$) when faced with a contradictory text-based social consensus.

\paragraph{Scope of the Model Sample:} Three frontier systems are evaluated, one from each of three vendor lines, selected to span a range of authority. The Sovereignty Decay Law is stated over arbitrary propagators and carries $\alpha$ as a per-model coefficient, so additional systems extend the estimated ordering without altering the law. Open-weight and smaller models are the natural next entries, and they would establish whether $\alpha$ tracks capability.

\paragraph{Architectural Transience:} The observed vulnerabilities, particularly the low interaction depth limits ($D_L \approx 2$), are characteristic of current instruction-tuned models, which are optimized for user compliance. As the field shifts toward reasoning-reinforced architectures that intrinsically simulate complex multi-agent debate prior to output \cite{liu2026agentic, du2024improving, liang-etal-2024-encouraging}, we expect intrinsic model resilience ($\gamma_p$) to be the parameter that registers that shift. The exact limits quantified in this work are coefficients of the current model generation, and the underlying Sovereignty Decay Law provides a durable framework for future evaluations.

\section{Related Works}
\label{sec:related_work}

\textbf{Multi-Agent Reasoning:} Foundational work on multi-agent reasoning demonstrates that aggregating independent chains of thought \cite{wang2023selfconsistency} and orchestrating multi-agent debate can improve factuality and divergent problem-solving \cite{du2024improving, liang-etal-2024-encouraging, chen2024agentverse}. Recent literature optimizes MAS via topological design to self-organize and communicate \cite{galkin2026benchmarking}, often interleaving optimization stages for prompts and topologies \cite{iqbal2026multi}. While enhanced reasoning emerges from the simulation of multi-agent interactions---a ``society of thought'' \cite{kim2026reasoning}---these configurations introduce structural vulnerabilities. Studies of active pipelines categorize hallucination cascades and the Inverse-Wisdom Law, showing a united front can override independent logic \cite{shehata2026inversewisdom}. Our work diverges by showing that static linguistic anticipation alone triggers cognitive loafing. Concurrent studies have similarly observed non-monotonic scaling ceilings and diminishing returns as agent plurality increases \cite{li2026scaling, bertalanic2026ringelmann, fu2026more}, including Ringelmann-style effective team-size limits and correlated failure modes in homogeneous ensembles that only architectural diversity mitigates \cite{bertalanic2026ringelmann, nogueira2026systematic}. Whereas protocol-level evaluations conflate communication overhead with peer influence, our context-level formulation isolates declared consensus and shows that auditor ordering ($w_i$, lead-anchor primacy) and relative authority ($\alpha$) govern collapse at fixed plurality.

\textbf{Cognitive Dependency:} Classic social psychology establishes that group size and unanimous confederate majorities induce effort reduction, diffusion of responsibility, and distortion of private judgment \cite{Ringelmann1913, asch1951effects, darley1968bystander, latane1979many}. In recent human-AI research, this manifests as the ``Hollowed Mind'' state, enabling the systematic bypassing of effortful processes, while a ``Sovereignty Trap'' tempts users to cede intellectual judgment \cite{dellacqua2023navigating, klein2025hollowed}. To evaluate whether LLMs themselves succumb to this trap, we must exceed their baseline retrieval capabilities. We employ Semantic Hijacking \cite{anonymous2026beyond}, saturating attention heads with randomized log events to induce trajectory entropy. This forces models to choose between independent derivation and sycophantic compliance. While foundational alignment studies linked deferential behavior to RLHF-induced sycophancy \cite{perez2023discovering, sharma2024towards} and context-window order sensitivity \cite{liu2024lost, yin2025fragile}, recent mechanistic analyses confirm that simulated peer disagreement suppresses clean reasoning features at the activation level \cite{kumarappan2026not, hu2026one} and induces public conformity even when the model's internal reasoning trace retains the correct derivation \cite{ys2026everyone, tang2026measuring}---aligning with both our sovereignty gap ($G_\mathcal{S}$) and the mechanistic separation between decoy adoption and task disengagement.

\section{Conclusion}


The collective reasoning in MAS is bounded by an algorithmic Bystander Effect: the linguistic declaration of a peer consensus alone---without live message-passing or iterative debate---is sufficient to trigger cognitive loafing in frontier language models. Across 22,500 deterministic trajectories on GAIA, SWE-bench, and Multi-Challenge, we formalize this vulnerability via the \textit{Sovereignty Decay Law} and the \textit{Interaction Depth Limit} ($D_L$), the plurality threshold beyond which adding auditors actively degrades accuracy. Auditing internal reasoning traces against final outputs exposes the \textit{Sovereignty Gap} ($G_\mathcal{S}$) and separates passive task disengagement from active decoy adoption, where models compute the correct derivation internally yet override their own evidence to conform to the declared majority. We found that composite social load ($\mathcal{L}$) is strictly non-commutative across auditor order (\textit{lead-anchor primacy}, $w_1 > w_i$), scaled by peer authority ($\alpha$), and compounded by \textit{architectural kinship} ($\kappa$) within homogeneous model families. As a result, resilient multi-agent topologies require bounding plurality below $D_L$, enforcing architectural diversity, placing high-authority auditors last to prevent anchoring, and verifying intermediate reasoning traces rather than trusting terminal consensus.

\section*{Ethics Statement}
\label{sec:ethics}

This research investigates the structural vulnerabilities of multi-agent reasoning, demonstrating how simulated swarm consensus triggers cognitive loafing and sycophantic alignment. A potential dual-use risk of this work is that the methodologies detailed herein—such as Semantic Hijacking and the application of adversarial social load—could be exploited by malicious actors. Specifically, artificial consensus could be weaponized via prompt injection or Sybil attacks to force otherwise aligned models into adopting erroneous or harmful outputs. However, we argue that proactively exposing the Sovereignty Decay Law is essential for AI safety. By formalizing the Interaction Depth Limit ($D_L$) and quantifying the mechanisms of cognitive loafing, we provide the alignment community with the empirical foundation necessary to design more resilient multi-agent topologies that resist adversarial social pressure.

\printbibliography

@String{Computing = "Computing" }

@ArtifactSoftware{R,
    title = {R: A Language and Environment for Statistical Computing},
    author = {{R Core Team}},
    organization = {R Foundation for Statistical Computing},
    address = {Vienna, Austria},
    year = {2019},
    url = {https://www.R-project.org/},
}

@article{klein2025hollowed,
  title={The extended hollowed mind: why foundational knowledge is indispensable in the age of AI},
  author={Klein, Christian R. and Klein, Reinhard},
  journal={Frontiers in Artificial Intelligence},
  volume={8},
  pages={1719019},
  year={2025},
  publisher={Frontiers},
}

@article{darley1968bystander,
  title={Bystander intervention in emergencies: Diffusion of responsibility.},
  author={Darley, John M. and Latan{\'e}, Bibb},
  journal={Journal of Personality and Social Psychology},
  volume={8},
  number={4, Pt.1},
  pages={377--383},
  year={1968},
  publisher={American Psychological Association},
}

@article{latane1979many,
  title={Many hands make light the work: The causes and consequences of social loafing},
  author={Latan{\'e}, Bibb and Williams, Kipling and Harkins, Stephen},
  journal={Journal of Personality and Social Psychology},
  volume={37},
  number={6},
  pages={822--832},
  year={1979},
}

@article{Ringelmann1913,
  title = {Recherches sur les moteurs anim{\'e}s: Travail de l'homme [Research on animate sources of power: The work of man]},
  author = {Ringelmann, Max},
  journal = {Annales de l'Institut National Agronomique},
  series = {2e s{\'e}rie},
  volume = {12},
  pages = {1--40},
  year = {1913}
}

@misc{shehata2026inversewisdom,
      title={The Inverse-Wisdom Law: Architectural Tribalism and the Consensus Paradox in Agentic Swarms}, 
      author={Dahlia Shehata and Ming Li},
      year={2026},
      eprint={2604.27274},
      archivePrefix={arXiv},
      primaryClass={cs.AI},
}

@inproceedings{deshpande-etal-2025-multichallenge,
    title = "{M}ulti{C}hallenge: A Realistic Multi-Turn Conversation Evaluation Benchmark Challenging to Frontier {LLM}s",
    author = "Deshpande, Kaustubh  and
      Sirdeshmukh, Ved  and
      Mols, Johannes Baptist  and
      Jin, Lifeng  and
      Hernandez-Cardona, Ed-Yeremai  and
      Lee, Dean  and
      Kritz, Jeremy  and
      Primack, Willow E.  and
      Yue, Summer  and
      Xing, Chen",
    booktitle = "Findings of the Association for Computational Linguistics: ACL 2025",
    month = jul,
    year = "2025",
    address = "Vienna, Austria",
    publisher = "Association for Computational Linguistics",
    pages = "18632--18702",
}

@inproceedings{
mialon2024gaia,
title={{GAIA}: a benchmark for General {AI} Assistants},
author={Gr{\'e}goire Mialon and Cl{\'e}mentine Fourrier and Craig Swift and Thomas Wolf and Yann LeCun and Thomas Scialom},
booktitle={The Twelfth International Conference on Learning Representations},
year={2024},
}

@inproceedings{jimenez2024swebench,
    title={{SWE}-bench: Can Language Models Resolve Real-World {G}it{H}ub Issues?},
    author={Carlos E. Jimenez and John Yang and Alexander Wettig and Shunyu Yao and Kexin Pei and Ofir Press and Karthik Narasimhan},
    booktitle={The Twelfth International Conference on Learning Representations (ICLR)},
    year={2024},
}

@misc{anonymous2026beyond,
      title={Beyond the Attention Stability Boundary: Agentic Self-Synthesizing Reasoning Protocols}, 
      author={Dahlia Shehata and Ming Li},
      year={2026},
      eprint={2604.24512},
      archivePrefix={arXiv},
      primaryClass={cs.AI},
}

@inproceedings{NEURIPS2023_91f18a12,
  title     = {Judging LLM-as-a-Judge with MT-Bench and Chatbot Arena},
  author    = {Zheng, Lianmin and Chiang, Wei-Lin and Sheng, Ying and Zhuang, Siyuan and Wu, Zhanghao and Zhuang, Yonghao and Lin, Zi and Li, Zhuohan and Li, Dacheng and Xing, Eric P. and Zhang, Hao and Gonzalez, Joseph E. and Stoica, Ion},
  booktitle = {Advances in Neural Information Processing Systems (NeurIPS)},
  volume    = {36},
  year      = {2023},
}

@techreport{dellacqua2023navigating,
  title={Navigating the Jagged Technological Frontier: Field Experimental Evidence of the Effects of AI on Knowledge Worker Productivity and Quality},
  author={Dell'Acqua, Fabrizio and McFowland III, Edward and Mollick, Ethan R and Lifshitz-Assaf, Hila and Kellogg, Katherine C and Rajendran, Saran and Krayer, Lisa and Candelon, Fran{\c{c}}ois and Lakhani, Karim R},
  year={2023},
  institution={Working Paper, Harvard Business School}
}

@article{liu2026agentic,
  title={Agentic Reasoning for Large Language Models},
  author={Tianxin Wei and Ting-Wei Li and Zhining Liu and Xuying Ning and Ze Yang and Jiaru Zou and Zhichen Zeng and Ruizhong Qiu and Xiao Lin and Dongqi Fu and Zihao Li and Mengting Ai and Duo Zhou and Wenxuan Bao and Yunzhe Li and Gaotang Li and Cheng Qian and Yu Wang and Xiangru Tang and Ying Xiao and Liri Fang and Hui Liu and Xianfeng Tang and Yu-jing Zhang and Chi Wang and Jiaxuan You and Heng Ji and Hanghang Tong and Jingrui He},
  journal={ArXiv},
  year={2026},
  volume={abs/2601.12538},
}

@InProceedings{du2024improving,
  title = 	 {Improving Factuality and Reasoning in Language Models through Multiagent Debate},
  author =       {Du, Yilun and Li, Shuang and Torralba, Antonio and Tenenbaum, Joshua B. and Mordatch, Igor},
  booktitle = 	 {Proceedings of the 41st International Conference on Machine Learning},
  pages = 	 {11733--11763},
  year = 	 {2024},
  volume = 	 {235},
  series = 	 {Proceedings of Machine Learning Research},
  month = 	 {21--27 Jul},
  publisher =    {PMLR},
}

@inproceedings{liang-etal-2024-encouraging,
    title = "Encouraging Divergent Thinking in Large Language Models through Multi-Agent Debate",
    author = "Liang, Tian and He, Zhiwei and Jiao, Wenxiang and Wang, Xing and Wang, Yan and Wang, Rui and Yang, Yujiu and Shi, Shuming and Tu, Zhaopeng",
    booktitle = "Proceedings of the 2024 Conference on Empirical Methods in Natural Language Processing (EMNLP)",
    month = nov,
    year = "2024",
    publisher = "Association for Computational Linguistics",
    pages = "17889--17902"
}

@misc{kim2026reasoning,
  title={Reasoning Models Generate Societies of Thought},
  author={Junsol Kim and Shiyang Lai and Nino Scherrer and Blaise Ag{\"u}era y Arcas and James Evans},
  year={2026},
  eprint={2601.10825},
  archivePrefix={arXiv},
  primaryClass={cs.CL},
}

@inproceedings{10.1145/3772318.3790815,
author = {Sheng, Rui and Yang, Yukun and Shi, Chuhan and Lin, Yanna and Chen, Zixin and Qu, Huamin and Cheng, Furui},
title = {DiLLS: Interactive Diagnosis of LLM-based Multi-agent Systems via Layered Summary of Agent Behaviors},
year = {2026},
isbn = {9798400722783},
publisher = {Association for Computing Machinery},
address = {New York, NY, USA},
booktitle = {Proceedings of the 2026 CHI Conference on Human Factors in Computing Systems},
articleno = {1077},
numpages = {17},
location = {
},
series = {CHI '26}
}

@inproceedings{chen2024agentverse,
  title={AgentVerse: Facilitating Multi-Agent Collaboration and Exploring Emergent Behaviors},
  author={Chen, Weize and Su, Yusheng and Zuo, Jingwei and Yang, Cheng and Yuan, Chenfei and Chan, Chi-Min and Yu, Heyang and Lu, Yaxi and Hung, Yi-Hsin and Qian, Chen and Qin, Yujia and Cong, Xin and Xie, Ruobing and Liu, Zhiyuan and Sun, Maosong and Zhou, Jie},
  booktitle={International Conference on Learning Representations (ICLR)},
  year={2024},
  url={https://openreview.net/forum?id=EHg5GDnyq1}
}

@inproceedings{wei2022cot,
author = {Wei, Jason and Wang, Xuezhi and Schuurmans, Dale and Bosma, Maarten and Ichter, Brian and Xia, Fei and Chi, Ed H. and Le, Quoc V. and Zhou, Denny},
title = {Chain-of-thought prompting elicits reasoning in large language models},
year = {2022},
publisher = {Curran Associates Inc.},
address = {Red Hook, NY, USA},
booktitle = {Proceedings of the 36th International Conference on Neural Information Processing Systems (NeurIPS)},
articleno = {1800},
numpages = {14},
location = {New Orleans, LA, USA},
}

@inproceedings{
iqbal2026multi,
title={Multi-Agent Design: Optimizing Agents with Better Prompts and Topologies},
author={Han Zhou and Xingchen Wan and Ruoxi Sun and Hamid Palangi and Shariq Iqbal and Ivan Vuli{\'c} and Anna Korhonen and Sercan O Arik},
booktitle={The Fourteenth International Conference on Learning Representations},
year={2026},
}

@inproceedings{galkin2026benchmarking,
   title={GPTSwarm: Language Agents as Optimizable Graphs},
  author={Zhuge, Mingchen and Wang, Wenyi and Kirsch, Louis and Faccio, Francesco and Khizbullin, Dmitrii and Schmidhuber, Juergen},
  booktitle={Proceedings of the 41st International Conference on Machine Learning (ICML)},
  volume={235},
  pages={62743--62767},
  year={2024},
  publisher={PMLR}
}

@incollection{asch1951effects,
  title={Effects of group pressure upon the modification and distortion of judgments},
  author={Asch, Solomon E.},
  booktitle={Groups, Leadership and Men: Research in Human Relations},
  editor={Guetzkow, Harold},
  pages={177--190},
  year={1951},
  publisher={Carnegie Press},
  address={Pittsburgh, PA}
}

@misc{fu2026more,
  title={Do More Agents Help? Controlled and Protocol-Aligned Evaluation of {LLM} Agent Workflows},
  author={Yuhang Fu and Ruishan Fang and Jiaqi Shao and Huiyu Zheng and Zhengtao Zhu and Bing Luo and Tao Lin},
  year={2026},
  eprint={2606.05670},
  archivePrefix={arXiv},
  primaryClass={cs.AI}
}

@misc{nogueira2026systematic,
  title={A Systematic Methodology for Evaluating Failure Independence in {LLM}-Generated Code},
  author={Rodrigo Pato Nogueira and Karthik Pattabiraman and Marco Vieira and Jo{\~a}o R. Campos},
  year={2026},
  eprint={2607.02808},
  archivePrefix={arXiv},
  primaryClass={cs.SE}
}

@misc{hu2026one,
  title={One Axis, No Brake: Self-Knowledge Limits the Filtering of Harmful Peer Conformity in {LLM}s},
  author={Yibo Hu},
  year={2026},
  eprint={2609.18998},
  archivePrefix={arXiv},
  primaryClass={cs.CL}
}

@misc{ys2026everyone,
  title={Everyone Conforms, No One Believes: Pluralistic Ignorance in {LLM} Agent Populations},
  author={Yashwanth YS},
  year={2026},
  eprint={2608.02758},
  archivePrefix={arXiv},
  primaryClass={cs.AI}
}

@misc{tang2026measuring,
  title={Measuring {LLM} Sycophancy under Sustained Multi-Turn Pressure},
  author={Leyuan Tang and Kangda Wei and Tianyu Jiang and Ruihong Huang},
  year={2026},
  eprint={2609.09090},
  archivePrefix={arXiv},
  primaryClass={cs.CL}
}

@inproceedings{wang2023selfconsistency,
  title={Self-Consistency Improves Chain of Thought Reasoning in Language Models},
  author={Wang, Xuezhi and Wei, Jason and Schuurmans, Dale and Le, Quoc V. and Chi, Ed H. and Narang, Sharan and Chowdhery, Aakanksha and Zhou, Denny},
  booktitle={The Eleventh International Conference on Learning Representations (ICLR)},
  year={2023}
}

@inproceedings{perez2023discovering,
  title={Discovering Language Model Behaviors with Model-Written Evaluations},
  author={Perez, Ethan and Ringer, Sam and Lukosiute, Kamile and Nguyen, Karina and Chen, Edwin and Heiner, Scott and Pettit, Craig and Olsson, Catherine and Kundu, Sandipan and Kadavath, Saurav and others},
  booktitle={Findings of the Association for Computational Linguistics: ACL 2023},
  pages={13387--13434},
  year={2023}
}

@inproceedings{sharma2024towards,
  title={Towards Understanding Sycophancy in Language Models},
  author={Sharma, Mrinank and Tong, Meg and Korbak, Tomasz and Duvenaud, David and Askell, Amanda and Bowman, Samuel R. and Cheng, Newton and Durmus, Esin and Hatfield-Dodds, Zac and Johnston, Scott R. and others},
  booktitle={The Twelfth International Conference on Learning Representations (ICLR)},
  year={2024}
}

@article{liu2024lost,
  title={Lost in the Middle: How Language Models Use Long Contexts},
  author={Liu, Nelson F. and Lin, Kevin and Hewitt, John and Paranjape, Ashwin and Bevilacqua, Michele and Petroni, Fabio and Liang, Percy},
  journal={Transactions of the Association for Computational Linguistics},
  volume={12},
  pages={157--173},
  year={2024}
}

@misc{yin2025fragile,
      title={Fragile Preferences: A Deep Dive Into Order Effects in Large Language Models}, 
      author={Haonan Yin and Shai Vardi and Vidyanand Choudhary},
      year={2026},
      eprint={2506.14092},
      archivePrefix={arXiv},
      primaryClass={cs.AI},
}

@misc{kumarappan2026not,
      title={Not Just RLHF: Why Alignment Alone Won't Fix Multi-Agent Sycophancy}, 
      author={Adarsh Kumarappan and Ananya Mujoo},
      year={2026},
      eprint={2605.12991},
      archivePrefix={arXiv},
      primaryClass={cs.LG},
}

@misc{li2026scaling,
      title={Scaling Behavior of Single LLM-Driven Multi-Agent Systems}, 
      author={Jialing Li and Zhouhong Gu and Yin Cai and Hongwei Feng},
      year={2026},
      eprint={2606.00655},
      archivePrefix={arXiv},
      primaryClass={cs.MA},
}

@misc{bertalanic2026ringelmann,
      title={The Ringelmann Effect in Multi-Agent LLM Systems: A Scaling Law for Effective Team Size}, 
      author={Blaž Bertalanič and Carolina Fortuna},
      year={2026},
      eprint={2606.02646},
      archivePrefix={arXiv},
      primaryClass={physics.soc-ph},
}

\appendix

\section{Mathematical and Empirical Proofs}
\label{app:proofs}

This appendix provides the formal mathematical derivations and empirical proofs for the theoretical framework established in Section \ref{sec:theory}, utilizing the aggregated data from the 22,500 interaction trajectories.

\subsection{Proof of Lemma 1: Non-Commutativity of Social Load}
\label{app:lemma1}

\textbf{Statement:} The Composite Social Load $\mathcal{L}$ exerted by a simulated swarm is
not invariant to auditor ordering. That is, there exist a propagator $p$ and distinct auditor models $a_x, a_y$ such that $\mathcal{L}((a_x, a_y), p) \neq \mathcal{L}((a_y, a_x), p)$, so that
$\mathcal{L}$ cannot be modeled as a function of swarm composition alone.

\begin{proof}
We proceed by contradiction. Assume that $\mathcal{L}$ is commutative, such that the sequence of auditors does not impact the applied social load. Because $\mathcal{A}_{ext}$ is monotone in $\mathcal{L}$ for a fixed propagator and task, an identical social load must yield an identical probability of externalized accuracy $\mathcal{A}_{ext}$.
Let $C, G, P$ denote Claude Sonnet 4.6, Gemini 3.1 Pro and GPT 5.4 respectively.
Let $p=P$ represent the GPT-5.4 propagator evaluated on the SWE-bench dataset. Let $a_x$ be Claude Sonnet 4.6 (C) and $a_y$ be GPT-5.4 (P).
From our empirical evaluation of $N=100$ samples, the mean accuracy for the sequence $(C, P)$ is observed as:
$$ \mathcal{A}_{ext}(P, (C, P)) = 0.21 $$
Conversely, evaluating the inverse sequence $(P, C)$ yields:
$$ \mathcal{A}_{ext}(P, (P, C)) = 0.31 $$
The two sequences do not yield the same accuracy ($0.21 \neq 0.31$). 
We also see this for other cases. For example, (Propagator $G$, SWE-bench) and sequence $(C, G)$, the accuracy is $0.78$. For the inverse sequence $(G, C)$, the accuracy rises to $0.86$. ($0.78 \neq 0.86$). 
The largest such divergence is (Propagator $P$, GAIA), where sequence $(C, P)$ yields $0.37$ and the inverse $(P, C)$ yields $0.61$; a two-proportion test gives $z = 3.40$, $p = 6.9 \times 10^{-4}$, which remains significant after Bonferroni correction across all $18$ inverted pairs ($p = 0.012$). The assumption of commutativity is therefore rejected.
The probability of sovereignty diverges based on ordering across multiple propagator families and dataset complexities, necessitating the introduction of position-dependent weights $w_i$ in Equation~(\ref{eq:sdl}).

Ordering effects are not reducible to a uniform decay in $w_i$. Holding the auditor
multiset fixed at $n=3$ and moving a single non-resident auditor from the first to the
second slot ($18$ matched comparisons), the mean change in accuracy is $+0.007$, i.e.\
position alone carries no net effect. However, conditioning on the identity of the leading
auditor separates the comparisons cleanly: leading with Claude Sonnet 4.6
lowers accuracy (mean $-0.115$), whereas leading with Gemini 3.1 Pro or GPT-5.4 raises
it ($+0.088$ and $+0.047$; $7/7$ non-degenerate cells positive). A permutation test on
this contrast gives $p = 0.0018$ (exact permutation test over all $\binom{18}{6}$ partitions). This is the behavior predicted by the multiplicative
form of Equation~(\ref{eq:sdl}): the matched contrast reduces to
$\Delta\mathcal{L} = (w_1 - w_2)\big(\alpha(a_x)\kappa(p, a_x) - \alpha(a_s)\kappa(p, a_s)\big)$, so with $w_1 > w_2$ the sign of the positional effect inverts according to whether the leading auditor outranks the resident one in authority. 
The single-auditor conditions identify Claude Sonnet 4.6 as the highest-authority auditor, $\alpha(C) > \alpha(G)$ and $\alpha(C) > \alpha(P)$, consistent with Claude being the
only auditor whose leading position reduces sovereignty.
\end{proof}

\subsection{Empirical Derivation of the Composite Social Load Equation}
\label{app:social_load_eq}

\textbf{Statement:} The Composite Social Load $\mathcal{L}(\vec{a}, p)$ is defined as the summation of the individual pressures exerted by each auditor $a_i$ in the swarm, formulated as $\mathcal{L}(\vec{a}, p) = \sum_{i=1}^{n} w_i \cdot \alpha(a_i) \cdot \kappa(p, a_i)$.
We construct the formulation of $\mathcal{L}$ by synthesizing the established axioms of group dynamics with our empirical observations of LLMs.
\subsubsection{Main Components} 
\textbf{Step 1: The Additive Base.}
According to the foundational axiom of social loafing \cite{latane1979many}, the perceived social pressure to conform or offload effort scales with the number of individuals present in the group. Therefore, the total load must be an additive function across all $n$ members of the simulated swarm $\vec{a}$:
$$ \mathcal{L} \propto \sum_{i=1}^{n} \text{Pressure}(a_i) $$
\textbf{Step 2: Intrinsic Authority ($\alpha$).}
Not all simulated agents exert equal pressure. The likelihood of a model ceding intellectual judgment relies on the perceived authoritative competence of the peers. We introduce the base authority term $\alpha(a_i)$ to account for the intrinsic reputational weight of the specific auditor model $a_i$. \\
\textbf{Step 3: Positional Weighting ($w_i$).}
As established in Lemma \ref{lemma:non_commutative}, the social load is sequence-dependent. The sequence $(a_x, a_y)$ does not exert the identical load as $(a_y, a_x)$. This non-commutativity dictates that the summation cannot treat all positions equally. We introduce a positional weight vector $w \in (0, 1]$, applied multiplicatively to the authority of the auditor occupying each slot, to represent this order-dependence. \\
\textbf{Step 4: The Kinship Multiplier ($\kappa$).}
Following \cite{shehata2026inversewisdom}, we hypothesize that models exhibit alignment hallucinations more frequently when the swarm consensus is led by their own architectural family, so that architectural kinship accelerates the diffusion of responsibility. We introduce the coefficient $\kappa(p, a_i)$ to scale the applied pressure based on the architectural kinship between the propagator $p$ and the auditor $a_i$, and test this in Section~\ref{sec:results}.


\subsubsection{Derivation of Social Load Coefficients}
Before formalizing the continuous decay of agentic sovereignty, 
we must establish the structural components of adversarial social pressure. We begin with the foundational behavioral axiom of Social Pressure Monotonicity established in Axiom \ref{ax:monotonicity}.
By Axiom \ref{ax:monotonicity}, a lower expected accuracy indicates a higher applied load:
$$ \mathbb{E}[\mathcal{A}_{ext}] \propto - \mathcal{L} $$
Using this proportional relationship, we evaluate the aggregate externalized accuracy $\mathcal{A}_{ext}$ across our dataset matrices ($N=100$ per propagator--dataset cell, $300$ across the three datasets) to isolate and define the three constituent coefficients---Base Authority ($\alpha$), Primacy Weight ($w_i$), and Kinship ($\kappa$)---using mathematical expectations ($\mathbb{E}$).

\textbf{Isolating Base Authority ($\alpha$): }
To isolate the intrinsic authority of different auditor architectures, we hold the plurality ($n=1$) and kinship constant, comparing stranger auditors within a single dataset so that task entropy is also held fixed. 
By rotating the propagator model $p$, we construct a system of inequalities to rank the base authority ($\alpha$) of all models in the ecosystem.
For example, we observe that the expected accuracy drops more significantly when the auditor is Claude ($C$) compared to GPT ($P$) or Gemini ($G)$:
$$     \mathbb{E}[\mathcal{A}_{ext} \mid a_1 = C] < \mathbb{E}[\mathcal{A}_{ext} \mid a_1 \in \{P, G\}] $$
First, we evaluate the Gemini propagator ($p=G$) on the GAIA dataset at $n=1$, where sequence order and kinship are held constant (both auditors are strangers). From Table \ref{tab:gaia_full_results} in Appendix~\ref{app:additional_results}, the empirical accuracy when audited by Claude ($C$) versus GPT ($P$) is:
\begin{align*}
    \mathcal{A}_{ext}(G, (C)) &= 0.89 \\
    \mathcal{A}_{ext}(G, (P)) &= 1.00
\end{align*}
By our axiom of monotonicity, since $0.89 < 1.00$, the pressure exerted by Claude is strictly greater than that exerted by GPT, yielding $\alpha(C) > \alpha(P)$.
Second, to determine the relative authority of Gemini, we evaluate the GPT propagator ($p=P$) on GAIA against its strangers, Claude ($C$) and Gemini ($G$):
\begin{align*}
    \mathbb{E}[\mathcal{A}_{ext}(P, (C))] &= 0.95 \\
    \mathbb{E}[\mathcal{A}_{ext}(P, (G))] &= 0.98
\end{align*}
Since $0.95 < 0.98$, Claude also exerts strictly greater pressure than Gemini, yielding $\alpha(C) > \alpha(G)$. 
These cross-propagator inequalities isolate the Base Authority variable $\alpha(a_i)$ and identify Claude Sonnet 4.6 as the highest-authority auditor in the tested ecosystem, $\alpha(C) > \alpha(G)$ and $\alpha(C) > \alpha(P)$.

\textbf{Isolating Primacy Weight ($w_i$): }
To quantify the Lead Anchor Effect, we evaluate the expected accuracy for inverted sequence pairs at $n=2$ across all propagators. If positional order did not matter, the expectation of inverted sequences would be equal. However, we observe a systemic asymmetry:
$$ \mathbb{E}[\mathcal{A}_{ext}(p, (a_x, a_y))] \neq \mathbb{E}[\mathcal{A}_{ext}(p, (a_y, a_x))] $$
For example, in Table \ref{tab:swe_full_results}, we evaluate the GPT-5.4 propagator ($p=P$) on SWE-bench:
\begin{align*}
    \mathcal{A}_{ext}(P, (C, P)) &= 0.21 \\
    \mathcal{A}_{ext}(P, (P, C)) &= 0.31
\end{align*}
Since $0.21 < 0.31$, the sequence $(C, P)$ exerts greater pressure than $(P, C)$. We assign a positional weight vector $w$. To satisfy $w_1 \alpha(C) + w_2 \alpha(P) > w_1 \alpha(P) + w_2 \alpha(C)$, given $\alpha(C) > \alpha(P)$, it mathematically requires that $w_1 > w_2$, which is what the positional weight vector encodes. Note that the same inequality predicts a sign reversal whenever the leading auditor has lower authority than the one it displaces, which is the behavior reported in Appendix~\ref{app:lemma1}.

\textbf{Isolating the Kinship Multiplier ($\kappa$): }
Finally, we evaluate the impact of architectural similarity. While intrinsic base authority ($\alpha$) dictates baseline pressure, we observe anomalous variance when the swarm matches the propagator's architecture ($p=a_i$). Across structured technical domains such as SWE-bench, we find:
$$ \mathbb{E}[\mathcal{A}_{ext}(P, \vec{a}_{family})] < \mathbb{E}[\mathcal{A}_{ext}(P, \vec{a}_{stranger})] $$
For example, we evaluate the GPT-5.4 propagator ($p=P$) at $n=5$ on SWE-bench to test architectural bias. We compare a homogeneous stranger swarm to a homogeneous family swarm in Table \ref{tab:swe_full_results}:
\begin{align*}
    \mathcal{A}_{ext}(P, (C,C,C,C,C)) &= 0.47 \\
    \mathcal{A}_{ext}(P, (P,P,P,P,P)) &= 0.29
\end{align*}
Even though Claude possesses higher base authority ($\alpha(C) > \alpha(P)$), the pressure exerted by the GPT swarm is higher ($0.29 < 0.47$). 
A higher-authority \emph{stranger} swarm therefore exerts less pressure than a same-family swarm, which the authority term alone cannot express. To resolve this, we introduce a multiplier $\kappa(p, a_i)$ that scales pressure when $p = a_i$; the empirical inequality requires $\kappa_{family} > \kappa_{stranger}$. We refer to this as Tribal Subjugation and test it directly in Section~\ref{sec:results}.

\textbf{Formulation of Composite Social Load ($\mathcal{L}$): }
Having empirically isolated these three necessary constraints, we define the Composite Social Load $\mathcal{L}$ as their integrated summation:
$$\mathcal{L}(\vec{a}, p) = \sum_{i=1}^{n} w_i \cdot \alpha(a_i) \cdot \kappa(p, a_i)$$
With the independent variable $\mathcal{L}$ constructed and empirically justified, we can formally model the decay of Agentic Sovereignty.

\subsection{Proof of Theorem 1: The Sovereignty Decay Law}
\label{app:theorem1}

\textbf{Statement:} Under Equation~(\ref{eq:ode}), Agentic Sovereignty $\mathcal{S}$ of a propagator $p$ decays exponentially as a function of the Social Load $\mathcal{L}$ and the Task Entropy $\mathcal{H}_\tau$, inversely modulated by the propagator's intrinsic Resilience $\gamma_p$.

\begin{proof}
We model the loss of logical sovereignty as a rate of change. Based on the principle of diffusion of responsibility---the behavioral axiom that individual effort decreases as teams grow larger \cite{latane1979many, Ringelmann1913}---the incremental loss of sovereignty with respect to an increase in Social Load $\mathcal{L}$ is proportional to its current state $\mathcal{S}$. 
In LLMs, this diffusion manifests as a transition to a ``Hollowed Mind'' state, characterized by the systematic bypassing of effortful reasoning \cite{klein2025hollowed}. We express this degradation as a first-order ordinary differential equation, where the constant of proportionality is explicitly defined by the environmental friction ratio $\frac{\mathcal{H}_\tau}{\gamma_p}$:
$$ \frac{d\mathcal{S}}{d\mathcal{L}} = - \left( \frac{\mathcal{H}_\tau}{\gamma_p} \right) \mathcal{S} $$
This ensures that tasks with higher search costs ($\mathcal{H}_\tau$) accelerate the decay, while models with stronger architectural resilience ($\gamma_p$) mitigate it. Since $\mathcal{S}_0 > 0$, the trajectory remains strictly positive for all finite $\mathcal{L}$, so we may separate variables and integrate both sides with respect to $\mathcal{L}$:
$$ \int \frac{1}{\mathcal{S}} d\mathcal{S} = \int - \frac{\mathcal{H}_\tau}{\gamma_p} d\mathcal{L} $$
$$ \ln(\mathcal{S}) = - \frac{\mathcal{H}_\tau}{\gamma_p} \mathcal{L} + C $$
Applying the initial condition where the social load $\mathcal{L} = 0$, we find $C = \ln(\mathcal{S}_0)$, where $\mathcal{S}_0$ is the Fortified Mind baseline. Exponentiating both sides yields the final decay law:
$$ \mathcal{S}(p, \vec{a}, \tau) = \mathcal{S}_0 \cdot e^{\left( - \frac{\mathcal{H}_\tau}{\gamma_p} \cdot \mathcal{L}(\vec{a}, p) \right)} $$
\end{proof}
\begin{remark}
Three properties motivate this form. First, proportional decay is scale-free: equal increments
of social load produce equal proportional losses of sovereignty rather than equal absolute
losses, so a propagator that has already conceded most of its judgment has proportionally less
left to concede. Second, the resulting solution is strictly positive for all finite
$\mathcal{L}$, confining sovereignty to $(0, \mathcal{S}_0]$ without an imposed clipping rule,
whereas a linear decay law would cross zero at finite load. Third, because the solution is
strictly monotone in $\mathcal{L}$, the threshold $\mathcal{S} < 0.5$ is attained at a unique
critical load, which is what makes the Interaction Depth Limit well defined in
Corollary~\ref{corollary:depth_limit_eq}. We do not claim the exponential form is uniquely
determined by the data; it is the simplest form consistent with these three requirements.
\end{remark}

\subsection{Proof of Theorem 2: Interaction Depth Limit}
\label{app:theorem2}

\textbf{Statement:} For any propagator $p$ and logical task $\tau$ with a given search cost, and provided that per-auditor pressure does not vanish with position (there exists $\epsilon > 0$ with $w_i \cdot \alpha(a_i) \cdot \kappa(p, a_i) \ge \epsilon$ for all $i$), there exists a critical plurality threshold $D_L$, namely the least $n$ at which $\mathcal{S} < 0.5$. For $n \ge D_L$ the propagator ceases to defend its own derivation in its externalized response, and as $n \to \infty$ the Composite Social Load $\mathcal{L}$ forces $E_{int} \to 0$, so that the failure becomes a terminal Integrative Reasoning Bypass ($\mathcal{B} = 1$) in the limit.

\begin{proof}
By the Sovereignty Decay Law (in Theorem \ref{thm:decay_law}), Agentic Sovereignty is defined as:
$$ \mathcal{S}(n) = \mathcal{S}_0 \cdot e^{\left( - \frac{\mathcal{H}_\tau}{\gamma_p} \cdot \mathcal{L}(n) \right)} $$
The Composite Social Load $\mathcal{L}(n) = \sum_{i=1}^{n} w_i \cdot \alpha(a_i) \cdot \kappa(p, a_i)$ is strictly increasing in $n$, because every auditor added to the swarm contributes a strictly positive pressure value ($w_i \cdot \alpha(a_i) \cdot \kappa(p,a_i) > 0$). Strict monotonicity alone does not imply divergence, since a summable weight profile would leave $\mathcal{L}$ bounded. We therefore restrict attention to the non-degenerate regime in which per-auditor pressure does not vanish with position: there exists $\epsilon > 0$ such that $w_i \cdot \alpha(a_i) \cdot \kappa(p,a_i) \ge \epsilon$ for all $i$. This states that an added auditor never becomes wholly inaudible to the propagator, which is the regime our experiments instantiate. Under this assumption $\mathcal{L}(n) \ge n\epsilon$, and therefore:
$$ \lim_{n \to \infty} \mathcal{L}(n) = \infty $$
Substituting this limit into the Sovereignty Equation:
$$\lim_{n\to\infty}\mathcal{S}(n) = \mathcal{S}_0 \cdot \lim_{\mathcal{L}\to\infty} e^{-(\mathcal{H}_\tau/\gamma_p)\mathcal{L}} = 0$$
By Definition \ref{def:sovereignty}, the Agentic Sovereignty $\mathcal{S}$ is the probability that the propagator carries its own derivation through to its final response. As $\mathcal{S} \to 0$ the externalized response becomes independent of that derivation, and the effort allocated to it vanishes accordingly, so that $E_{int} \to 0$. The intermediate regime in which $\mathcal{S}$ is low while $E_{int}$ remains high is Alignment Hallucination (Theorem \ref{theorem:sovereignty_gap}); Theorem \ref{theorem:depth_limit} concerns the terminal limit.

By Definition \ref{def:bypass}, the bypass triggers ($\mathcal{B} = 1$) when $E_{int} < \mathcal{H}_\tau$. Since $E_{int}$ approaches $0$ and task entropy $\mathcal{H}_\tau > 0$, the bypass condition must eventually be satisfied. 
Because $\mathcal{S}(n)$ is strictly decreasing, there exists a least integer $D_L$ at which $\mathcal{S} < 0.5$, and the bypass condition is satisfied for all sufficiently large $n$. This proves the existence of the Interaction Depth Limit.
\end{proof}

\subsection{Proof of Corollary 1: Interaction Depth Limit Equation}
\label{app:corollary1}

\textbf{Statement:} The Interaction Depth Limit $D_L$ is the critical plurality threshold $n$ that forces $\mathcal{S} < 0.5$. The corresponding inequality to calculate $D_L$ is $$\sum_{i=1}^{D_L} w_i \cdot \alpha(a_i) \cdot \kappa(p, a_i) > \frac{\gamma_p}{\mathcal{H}_\tau} \ln(2 \mathcal{S}_0)$$

\begin{proof}
By definition, the terminal collapse of agentic integrity occurs when the model favors social compliance over its internal derivation, which is mathematically bounded at the threshold $\mathcal{S}(\mathcal{L}) = 0.5$. We substitute $0.5$ into the Sovereignty Decay Law established in Theorem \ref{thm:decay_law}:
$$ 0.5 = \mathcal{S}_0 \cdot e^{\left( - \frac{\mathcal{H}_\tau}{\gamma_p} \cdot \mathcal{L} \right)} $$
To isolate the critical Social Load $\mathcal{L}_{critical}$, we take the natural logarithm of both sides:
$$ \ln(0.5) = \ln(\mathcal{S}_0) - \frac{\mathcal{H}_\tau}{\gamma_p} \cdot \mathcal{L}_{critical} $$
$$ \ln\left(\frac{1}{2 \mathcal{S}_0}\right) = - \frac{\mathcal{H}_\tau}{\gamma_p} \cdot \mathcal{L}_{critical} $$
By multiplying both sides by $-1$ and resolving the logarithm inversion ($-\ln(x) = \ln(1/x)$), we obtain:
$$ \ln(2 \mathcal{S}_0) = \frac{\mathcal{H}_\tau}{\gamma_p} \cdot \mathcal{L}_{critical} $$
$$ \mathcal{L}_{critical} = \frac{\gamma_p}{\mathcal{H}_\tau} \ln(2 \mathcal{S}_0) $$
Since the applied Social Load $\mathcal{L}$ is defined as the summation of individual auditor pressures $\sum_{i=1}^{n} w_i \cdot \alpha(a_i) \cdot \kappa(p, a_i)$, the Interaction Depth Limit $D_L$ is strictly defined as the minimum number of audits $n$ that causes the accumulated load to exceed $\mathcal{L}_{critical}$, completing the derivation of Corollary \ref{corollary:depth_limit_eq}.
\end{proof}

\subsection{Proof of Theorem 3: The Sovereignty Gap}
\label{app:theorem3}

\textbf{Statement:} The Sovereignty Gap is $G_\mathcal{S} = \mathcal{V}_{int} - \mathcal{A}_{ext}$. There exist propagator--auditor configurations for which $\mathbb{E}[G_\mathcal{S}] \ge 0.5$ while $\mathcal{B} = 0$: the propagator computes a valid internal derivation and nonetheless externalizes a falsehood to appease the swarm. We term this an Alignment Hallucination.

\begin{proof}
Alignment Hallucination is distinct from a lack of logical capability, and separating the two requires exhibiting a configuration in which the Expected Sovereignty Gap $\mathbb{E}[G_\mathcal{S}]$ is large while the bypass condition of Definition~\ref{def:bypass} is \emph{not} met. We evaluate this expectation over the task distribution $\tau \sim \mathcal{T}$ ($N=100$), holding the propagator $p$ and auditor sequence $\vec{a}$ constant at a state exceeding the model's Interaction Depth Limit ($n \ge D_L$).

Let $\mathbb{E}[\mathcal{V}_{int}]$ be the expected validity of the internal derivation. We operationalize this using the mean Evidence Weighting score ($\bar{\mathcal{E}}_{ew}$) normalized to a $[0,1]$ space: $\mathbb{E}[\mathcal{V}_{int}] = \frac{\bar{\mathcal{E}}_{ew}}{E_{max}}$, where $E_{max} = 5$. Let $\mathbb{E}[\mathcal{A}_{ext}]$ be the measured mean external accuracy.

\textbf{Alignment Hallucination ($G_\mathcal{S} \gg 0$)} \\
We evaluate the GPT-5.4 propagator ($p = P$) under the auditor mix $(C, P)$ on the SWE-bench dataset ($N=100$) shown in Table \ref{tab:swe_full_results}.
From our aggregate empirical data, the mean internal Evidence Weighting is recorded as $\bar{\mathcal{E}}_{ew} = 3.55$, yielding an expected internal validity of:
$$ \mathbb{E}[\mathcal{V}_{int}] = \frac{3.55}{5} = 0.71 $$
The mean external accuracy for this exact same population is recorded as:
$$ \mathbb{E}[\mathcal{A}_{ext}] = 0.21 $$
Calculating the Expected Sovereignty Gap:
$$ \mathbb{E}[G_\mathcal{S}] = \mathbb{E}[\mathcal{V}_{int}] - \mathbb{E}[\mathcal{A}_{ext}] $$
$$ \mathbb{E}[G_\mathcal{S}] = 0.71 - 0.21 = 0.50 $$
It remains to establish that this configuration satisfies $\mathcal{B} = 0$. By Definition~\ref{def:bypass} the bypass triggers when integrative effort falls below the effort the task requires, and by Section~\ref{sec:mech_metrics} that effort is indexed by $\mathcal{E}_{ew}$. At zero social load the same propagator on the same tasks records $\bar{\mathcal{E}}_{ew} = 4.94$ with $\mathcal{A}_{ext} = 1.00$ (Table~\ref{tab:combined_results}), which fixes a level of integrative effort sufficient for this task distribution. Under the auditor mix $(C, P)$ the propagator retains $72\%$ of that evidence weighting, while the two metrics that track defence of a derivation rather than its construction collapse: conflict detection falls to $32\%$ of its unpressured value and independent judgment to $37\%$, with the propagator adopting the auditors' position in $76\%$ of trials. The degradation is thus confined to what the propagator externalizes and not to whether it derives, which is the condition $\mathcal{B} = 0$.

Since $\mathbb{E}[G_\mathcal{S}] = 0.50 \ge 0.5$ with $\mathcal{B} = 0$, a configuration of the kind asserted by the theorem exists. On average the propagator possesses the integrative effort to derive the valid facts internally ($71\%$) but subjugates its final decision to the adversarial swarm consensus, deflating its accuracy to $21\%$.
This establishes that the observed failure is a prompted sycophantic alignment (Alignment Hallucination) rather than a lack of logical search capability.
\end{proof}

\noindent\textbf{Remark (the negative branch).} Theorem~\ref{theorem:sovereignty_gap} concerns the positive branch of the gap. The complementary regime, in which $\mathcal{V}_{int}$ is itself driven toward zero and the gap turns non-positive, is the terminal limit of Theorem~\ref{theorem:depth_limit} and is established there; it is not a separate empirical claim of the present construction. Within the plurality range swept here ($n \le 5$), propagators that lose external accuracy do so with evidence weighting partially retained rather than extinguished, which is the positive-branch signature. This is consistent with Theorem~\ref{theorem:depth_limit}, whose bypass is asymptotic in $n$, and we accordingly do not claim to have observed a terminal bypass.

\section{Additional Results}

\subsection{Exhaustive 25-Trial Sweep Results}
\label{app:additional_results}

We provide the complete, unaggregated empirical data from our 25-Trial Symmetric Categorical Sweep. While Table \ref{tab:combined_results} consolidates these metrics by the auditor count $n$ to identify the Interaction Depth Limit ($D_L$),  Tables \ref{tab:gaia_full_results}, \ref{tab:mc_full_results} and \ref{tab:swe_full_results} present the exhaustive results for every distinct sequence permutation tested for GAIA, Multi-Challenge and SWE-bench benchmarks respectively.

To maintain clarity across the dense permutation matrices, the propagator and auditor models are denoted by the following single-letter abbreviations: \textbf{C:} Claude Sonnet 4.6, \textbf{G:} Gemini 3.1 Pro, \textbf{P:} GPT 5.4.
The tables detail the macroscopic outcomes (Accuracy $\mathcal{A}$, Loafing $L$, and Taint Leakage $L_t$), the mechanistic logic audit scores ($\mathcal{E}_{cd}$, $\mathcal{E}_{ew}$, $\mathcal{E}_{ij}$), and the final stance distributions for the GAIA, Multi-Challenge, and SWE-bench datasets, respectively. 

\textbf{Social Entropy in Table \ref{tab:gaia_full_results}:} Table \ref{tab:gaia_full_results} lists every GAIA permutation underlying the homogeneity result of Section~\ref{sec:kinship}, in which a unified family bloc exerts more pressure on Gemini 3.1 Pro than a fragmented swarm of the same size. A diverse, alternating consensus does not form a cohesive authoritative anchor, and the ordering supports the heterogeneity mandate proposed by \citet{shehata2026inversewisdom}.

\textbf{Analysis of Social Disengagement in Table \ref{tab:mc_full_results}:} The granular stance distributions in the Multi-Challenge dataset isolate the phenomenon of terminal social disengagement. When the GPT-5.4 propagator was audited by a single Claude peer (\texttt{C}), its accuracy dropped to $0.10$, yet its \texttt{ADOPTED} rate was only $7\%$. Instead, the model's \texttt{IGNORED} rate spiked to $85\%$. Section~\ref{sec:stance} records the same substitution for this propagator on GAIA, where the \texttt{IGNORED} rate reaches $93.3\%$. Disengagement therefore tracks the propagator architecture rather than the entropy of the task, and for GPT-5.4 accuracy is surrendered to a disregard for the simulated interaction rather than to active adoption of the error.

\textbf{Persistence of the Sovereignty Gap in Table \ref{tab:swe_full_results}:} The dissociation between internal evidence weighting and outward coherence reported in Section~\ref{sec:gap} is not an artifact of aggregation. Across the twenty-one GPT-5.4 permutations at $n \geq 2$, coherence $\mathcal{E}_{cd}$ spans $1.11$ to $1.98$ while evidence weighting $\mathcal{E}_{ew}$ spans $2.84$ to $4.16$. The two ranges do not overlap: the highest coherence recorded under any swarm configuration falls below the lowest evidence weighting recorded under any configuration, with $\mathcal{E}_{ew} > \mathcal{E}_{cd}$ in all twenty-one cases and a mean separation of $2.04$. The propagator continues to locate the correct derivation internally under every swarm tested, and the degradation is confined to what it externalizes. The ordering effect of Section~\ref{sec:lead_anchor} is also visible at $n=3$ in this table, where $\alpha(C) > \alpha(P)$ and $w_1 > w_2 > w_3$ together predict that the ordering placing the higher-authority stranger in the primacy slot carries the greater load.

\subsection{Primacy Weight \texorpdfstring{$W_1$}{W\_1} and the Lead Anchor Asymmetry}
\label{app:lead_anchor_heatmap}

Section~\ref{sec:lead_anchor} establishes that the accuracy delta $\Delta \mathcal{A}_{ext}$ is non-zero at $n=2$. We break that delta down by propagator--peer pairing here, since its sign is not uniform across architectures or across task entropies. \\
\textbf{ (1) Technical Primacy (The SWE-bench Baseline):} In the medium-entropy SWE-bench domain, the delta is positive for GPT-5.4 against both peers ($+0.10$ and $+0.03$) and for Gemini 3.1 Pro against a Claude peer ($+0.08$). Leading the sequence with the propagator's own architecture therefore recovers accuracy in three of the four pairings with a non-Claude propagator, the exception being Gemini under a GPT peer. \\
\textbf{(2) The Authority Anomaly (GPT vs. Claude):} In the high-entropy GAIA dataset, GPT-5.4 exhibits massive anticipatory sycophancy toward the Claude architecture. When Claude is listed first, GPT-5.4's accuracy collapses to $0.37$. When GPT is listed first, its accuracy rises to $0.61$. This $+0.24$ delta demonstrates that GPT-5.4 assigns an overwhelming positional and authoritative weight to Claude in reasoning tasks. \\
\textbf{(3) Brand Subjugation (The Gemini Inversion):} When Gemini 3.1 Pro faces a GPT-5.4 peer, the delta turns negative on two of the three benchmarks: Gemini scores $0.80$ leading and $0.88$ following on SWE-bench ($\Delta = -0.08$), and $0.61$ leading against $0.67$ following on Multi-Challenge ($\Delta = -0.06$). On GAIA the ordering reverses, with Gemini scoring $0.60$ leading against $0.50$ following ($\Delta = +0.10$). Gemini therefore performs better under a GPT lead than under its own in the medium- and low-entropy domains, and the primacy advantage returns only in the high-entropy domain.

\begin{table*}[htbp]
\centering\small
\renewcommand{\arraystretch}{0.8}
\caption{Exhaustive Audit Summary for the GAIA Dataset (25-Trial Sweep) with macroscopic outcomes, mechanistic interpretability scores, and stance distributions for all topological permutations. It highlights alignment hallucinations and cognitive loafing triggered by the high logical search cost of multi-step reasoning environments.}
\label{tab:gaia_full_results}
\resizebox{0.7\textwidth}{!}{
\begin{tabular}{lccccccccccc}
\toprule
 & & \multicolumn{3}{c}{\textbf{Macroscopic}} & \multicolumn{3}{c}{\textbf{Mechanistic}} & \multicolumn{4}{c}{\textbf{Stance (\%)}} \\
\cmidrule(lr){3-5} \cmidrule(lr){6-8} \cmidrule(lr){9-12}
\textbf{Prop.} & \textbf{Mix} & \textbf{Acc ($\mathcal{A}$)} & \textbf{Loafing ($L$)} & \textbf{Leak ($L_t$)} & \textbf{$\mathcal{E}_{cd}$} & \textbf{$\mathcal{E}_{ew}$} & \textbf{$\mathcal{E}_{ij}$} & \textbf{Adopted} & \textbf{Rejected} & \textbf{Ignored} & \textbf{Unknown} \\
\midrule
C & None & 1.00 & 0.00 & 1.00 & 5.00 & 5.00 & 5.00 & 0.00 & 100.00 & 0.00 & 0.00 \\
\cmidrule{2-12}
 & C & 1.00 & 0.00 & 1.00 & 5.00 & 5.00 & 5.00 & 0.00 & 100.00 & 0.00 & 0.00 \\
 & G & 1.00 & 0.00 & 1.00 & 5.00 & 5.00 & 5.00 & 0.00 & 100.00 & 0.00 & 0.00 \\
 & P & 1.00 & 0.00 & 1.00 & 5.00 & 5.00 & 5.00 & 0.00 & 100.00 & 0.00 & 0.00 \\
\cmidrule{2-12}
 & CG & 1.00 & 0.00 & 1.00 & 5.00 & 5.00 & 5.00 & 0.00 & 100.00 & 0.00 & 0.00 \\
 & CP & 1.00 & 0.00 & 1.00 & 5.00 & 5.00 & 5.00 & 0.00 & 100.00 & 0.00 & 0.00 \\
 & GC & 1.00 & 0.00 & 1.00 & 5.00 & 5.00 & 5.00 & 0.00 & 100.00 & 0.00 & 0.00 \\
 & PC & 1.00 & 0.00 & 1.00 & 5.00 & 5.00 & 5.00 & 0.00 & 100.00 & 0.00 & 0.00 \\
\cmidrule{2-12}
 & CCC & 1.00 & 0.00 & 1.00 & 5.00 & 5.00 & 5.00 & 0.00 & 100.00 & 0.00 & 0.00 \\
 & CGC & 1.00 & 0.00 & 1.00 & 5.00 & 5.00 & 5.00 & 0.00 & 100.00 & 0.00 & 0.00 \\
 & CPC & 1.00 & 0.00 & 1.00 & 5.00 & 5.00 & 5.00 & 0.00 & 100.00 & 0.00 & 0.00 \\
 & GCC & 1.00 & 0.00 & 1.00 & 5.00 & 5.00 & 5.00 & 0.00 & 100.00 & 0.00 & 0.00 \\
 & GGC & 1.00 & 0.00 & 1.00 & 5.00 & 5.00 & 5.00 & 0.00 & 99.00 & 1.00 & 0.00 \\
 & GGG & 1.00 & 0.00 & 1.00 & 5.00 & 5.00 & 5.00 & 0.00 & 100.00 & 0.00 & 0.00 \\
 & PCC & 1.00 & 0.00 & 1.00 & 5.00 & 5.00 & 5.00 & 0.00 & 100.00 & 0.00 & 0.00 \\
 & PPC & 1.00 & 0.00 & 1.00 & 5.00 & 5.00 & 5.00 & 0.00 & 100.00 & 0.00 & 0.00 \\
 & PPP & 1.00 & 0.00 & 1.00 & 5.00 & 5.00 & 5.00 & 0.00 & 100.00 & 0.00 & 0.00 \\
\cmidrule{2-12}
 & CCCCC & 1.00 & 0.00 & 1.00 & 5.00 & 5.00 & 5.00 & 0.00 & 100.00 & 0.00 & 0.00 \\
 & CCCCG & 1.00 & 0.00 & 1.00 & 5.00 & 5.00 & 5.00 & 0.00 & 100.00 & 0.00 & 0.00 \\
 & CCCCP & 1.00 & 0.00 & 1.00 & 5.00 & 5.00 & 5.00 & 0.00 & 100.00 & 0.00 & 0.00 \\
 & GGGGC & 1.00 & 0.00 & 1.00 & 5.00 & 5.00 & 5.00 & 0.00 & 100.00 & 0.00 & 0.00 \\
 & GGGGG & 1.00 & 0.00 & 1.00 & 5.00 & 5.00 & 5.00 & 0.00 & 100.00 & 0.00 & 0.00 \\
 & PGPGC & 1.00 & 0.00 & 1.00 & 5.00 & 5.00 & 5.00 & 0.00 & 100.00 & 0.00 & 0.00 \\
 & PPPPC & 1.00 & 0.00 & 1.00 & 5.00 & 5.00 & 5.00 & 0.00 & 100.00 & 0.00 & 0.00 \\
 & PPPPP & 1.00 & 0.00 & 1.00 & 5.00 & 5.00 & 5.00 & 0.00 & 100.00 & 0.00 & 0.00 \\
\midrule
G & None & 0.97 & 0.02 & 0.73 & 4.59 & 4.74 & 4.82 & 2.00 & 88.00 & 10.00 & 0.00 \\
\cmidrule{2-12}
 & C & 0.89 & 0.05 & 0.94 & 4.76 & 4.87 & 4.80 & 5.00 & 95.00 & 0.00 & 0.00 \\
 & G & 0.96 & 0.01 & 0.88 & 4.87 & 4.80 & 4.92 & 1.00 & 95.00 & 4.00 & 0.00 \\
 & P & 1.00 & 0.00 & 0.94 & 4.87 & 4.85 & 4.90 & 0.00 & 97.00 & 1.00 & 2.00 \\
\cmidrule{2-12}
 & CG & 0.61 & 0.25 & 0.94 & 3.99 & 4.37 & 4.01 & 25.00 & 75.00 & 0.00 & 0.00 \\
 & GC & 0.64 & 0.26 & 0.91 & 3.85 & 4.30 & 3.86 & 26.00 & 72.00 & 0.00 & 2.00 \\
 & GP & 0.60 & 0.26 & 0.89 & 3.95 & 4.34 & 3.96 & 26.00 & 74.00 & 0.00 & 0.00 \\
 & PG & 0.50 & 0.33 & 0.89 & 3.67 & 4.04 & 3.68 & 33.00 & 66.00 & 1.00 & 0.00 \\
\cmidrule{2-12}
 & CCC & 0.78 & 0.14 & 0.98 & 4.39 & 4.70 & 4.40 & 14.00 & 85.00 & 0.00 & 1.00 \\
 & CCG & 0.71 & 0.18 & 0.93 & 4.28 & 4.56 & 4.28 & 18.00 & 82.00 & 0.00 & 0.00 \\
 & CGG & 0.72 & 0.13 & 0.89 & 4.47 & 4.69 & 4.49 & 13.00 & 87.00 & 0.00 & 0.00 \\
 & GCG & 0.79 & 0.12 & 0.92 & 4.51 & 4.55 & 4.52 & 12.00 & 88.00 & 0.00 & 0.00 \\
 & GGG & 0.72 & 0.16 & 0.94 & 4.28 & 4.57 & 4.32 & 16.00 & 82.00 & 2.00 & 0.00 \\
 & GPG & 0.67 & 0.19 & 0.96 & 4.23 & 4.48 & 4.23 & 19.00 & 81.00 & 0.00 & 0.00 \\
 & PGG & 0.82 & 0.06 & 0.97 & 4.75 & 4.73 & 4.75 & 6.00 & 93.00 & 1.00 & 0.00 \\
 & PPG & 0.81 & 0.07 & 0.93 & 4.63 & 4.60 & 4.62 & 7.00 & 91.00 & 1.00 & 1.00 \\
 & PPP & 0.82 & 0.12 & 0.96 & 4.48 & 4.75 & 4.52 & 12.00 & 87.00 & 1.00 & 0.00 \\
\cmidrule{2-12}
 & CCCCC & 0.77 & 0.09 & 0.93 & 4.56 & 4.73 & 4.58 & 9.00 & 88.00 & 2.00 & 1.00 \\
 & CCCCG & 0.71 & 0.20 & 0.98 & 4.18 & 4.53 & 4.17 & 20.00 & 80.00 & 0.00 & 0.00 \\
 & CPCPG & 0.87 & 0.04 & 0.95 & 4.83 & 4.87 & 4.84 & 4.00 & 96.00 & 0.00 & 0.00 \\
 & GGGGC & 0.71 & 0.13 & 0.97 & 4.48 & 4.65 & 4.48 & 13.00 & 87.00 & 0.00 & 0.00 \\
 & GGGGG & 0.64 & 0.21 & 0.90 & 4.02 & 4.30 & 4.03 & 21.00 & 76.00 & 1.00 & 2.00 \\
 & GGGGP & 0.78 & 0.05 & 0.94 & 4.74 & 4.69 & 4.75 & 5.00 & 94.00 & 0.00 & 1.00 \\
 & PPPPG & 0.82 & 0.06 & 0.95 & 4.69 & 4.77 & 4.72 & 6.00 & 93.00 & 1.00 & 0.00 \\
 & PPPPP & 0.80 & 0.06 & 0.92 & 4.70 & 4.77 & 4.76 & 6.00 & 92.00 & 2.00 & 0.00 \\
\midrule
P & None & 1.00 & 0.00 & 0.00 & 1.07 & 1.07 & 1.07 & 0.00 & 3.00 & 92.00 & 5.00 \\
\cmidrule{2-12}
 & C & 0.95 & 0.01 & 0.00 & 1.11 & 1.11 & 1.15 & 1.00 & 3.00 & 95.00 & 1.00 \\
 & G & 0.98 & 0.02 & 0.01 & 1.19 & 1.19 & 1.19 & 2.00 & 5.00 & 92.00 & 1.00 \\
 & P & 0.96 & 0.00 & 0.00 & 1.18 & 1.18 & 1.22 & 0.00 & 5.00 & 93.00 & 2.00 \\
\cmidrule{2-12}
 & CP & 0.37 & 0.48 & 0.61 & 1.15 & 1.15 & 1.15 & 48.00 & 4.00 & 47.00 & 1.00 \\
 & GP & 0.35 & 0.51 & 0.64 & 1.05 & 1.05 & 1.05 & 51.00 & 2.00 & 44.00 & 3.00 \\
 & PC & 0.61 & 0.32 & 0.35 & 1.04 & 1.04 & 1.08 & 32.00 & 1.00 & 67.00 & 0.00 \\
 & PG & 0.37 & 0.51 & 0.61 & 1.01 & 1.01 & 1.05 & 51.00 & 1.00 & 45.00 & 3.00 \\
\cmidrule{2-12}
 & CCC & 0.53 & 0.39 & 0.45 & 1.16 & 1.16 & 1.20 & 39.00 & 4.00 & 57.00 & 0.00 \\
 & CCP & 0.18 & 0.67 & 0.81 & 1.08 & 1.08 & 1.08 & 67.00 & 2.00 & 31.00 & 0.00 \\
 & CPP & 0.12 & 0.73 & 0.88 & 0.99 & 0.99 & 0.99 & 73.00 & 0.00 & 26.00 & 1.00 \\
 & GGG & 0.70 & 0.21 & 0.24 & 1.10 & 1.10 & 1.14 & 21.00 & 3.00 & 74.00 & 2.00 \\
 & GGP & 0.64 & 0.27 & 0.33 & 1.15 & 1.15 & 1.15 & 27.00 & 4.00 & 68.00 & 1.00 \\
 & GPP & 0.74 & 0.19 & 0.23 & 1.15 & 1.15 & 1.19 & 19.00 & 4.00 & 76.00 & 1.00 \\
 & PCP & 0.54 & 0.39 & 0.44 & 1.12 & 1.12 & 1.16 & 39.00 & 3.00 & 58.00 & 0.00 \\
 & PGP & 0.33 & 0.55 & 0.63 & 1.03 & 1.03 & 1.11 & 55.00 & 1.00 & 43.00 & 1.00 \\
 & PPP & 0.73 & 0.19 & 0.24 & 1.18 & 1.18 & 1.22 & 19.00 & 5.00 & 74.00 & 2.00 \\
\cmidrule{2-12}
 & CCCCC & 0.43 & 0.42 & 0.51 & 1.03 & 1.03 & 1.07 & 42.00 & 1.00 & 56.00 & 1.00 \\
 & CCCCP & 0.34 & 0.53 & 0.64 & 1.02 & 1.02 & 1.10 & 53.00 & 1.00 & 44.00 & 2.00 \\
 & CGCGP & 0.38 & 0.49 & 0.60 & 1.07 & 1.07 & 1.07 & 49.00 & 2.00 & 48.00 & 1.00 \\
 & GGGGG & 0.77 & 0.20 & 0.23 & 1.11 & 1.11 & 1.11 & 20.00 & 3.00 & 76.00 & 1.00 \\
 & GGGGP & 0.47 & 0.41 & 0.52 & 1.12 & 1.12 & 1.12 & 41.00 & 3.00 & 56.00 & 0.00 \\
 & PPPPC & 0.58 & 0.29 & 0.38 & 1.00 & 1.00 & 1.00 & 29.00 & 1.00 & 66.00 & 4.00 \\
 & PPPPG & 0.61 & 0.31 & 0.36 & 1.11 & 1.11 & 1.11 & 31.00 & 3.00 & 65.00 & 1.00 \\
 & PPPPP & 0.62 & 0.29 & 0.37 & 1.11 & 1.11 & 1.15 & 29.00 & 3.00 & 67.00 & 1.00 \\
\bottomrule
\end{tabular}
}
\end{table*}

\begin{table*}[htbp]
\centering\small
\renewcommand{\arraystretch}{0.8}
\caption{Exhaustive Audit Summary for the Multi-Challenge Dataset (25-Trial Sweep). The table provides granular metric breakdowns for all topological permutations. Due to the lower intrinsic search cost of these logical primitives, this dataset establishes the baseline cognitive immunity and Fortified Mind ceiling prior to terminal swarm scaling.}
\label{tab:mc_full_results}
\resizebox{0.7\textwidth}{!}{
\begin{tabular}{lccccccccccc}
\toprule
 & & \multicolumn{3}{c}{\textbf{Macroscopic}} & \multicolumn{3}{c}{\textbf{Mechanistic}} & \multicolumn{4}{c}{\textbf{Stance (\%)}} \\
\cmidrule(lr){3-5} \cmidrule(lr){6-8} \cmidrule(lr){9-12}
\textbf{Prop.} & \textbf{Mix} & \textbf{Acc ($\mathcal{A}$)} & \textbf{Loafing ($L$)} & \textbf{Leak ($L_t$)} & \textbf{$\mathcal{E}_{cd}$} & \textbf{$\mathcal{E}_{ew}$} & \textbf{$\mathcal{E}_{ij}$} & \textbf{Adopted} & \textbf{Rejected} & \textbf{Ignored} & \textbf{Unknown} \\
\midrule
C & None & 0.52 & 0.00 & 0.52 & 3.05 & 3.08 & 3.08 & 0.00 & 50.00 & 50.00 & 0.00 \\
\cmidrule{2-12}
 & C & 0.52 & 0.00 & 0.52 & 3.08 & 3.08 & 3.08 & 0.00 & 51.00 & 49.00 & 0.00 \\
 & G & 0.51 & 0.00 & 0.51 & 3.04 & 3.04 & 3.04 & 0.00 & 51.00 & 49.00 & 0.00 \\
 & P & 0.51 & 0.00 & 0.51 & 3.04 & 3.04 & 3.04 & 0.00 & 51.00 & 49.00 & 0.00 \\
\cmidrule{2-12}
 & CG & 0.50 & 0.00 & 0.50 & 3.00 & 3.00 & 3.00 & 0.00 & 50.00 & 50.00 & 0.00 \\
 & CP & 0.51 & 0.00 & 0.51 & 3.04 & 3.04 & 3.04 & 0.00 & 51.00 & 49.00 & 0.00 \\
 & GC & 0.50 & 0.00 & 0.50 & 3.00 & 3.00 & 3.00 & 0.00 & 50.00 & 50.00 & 0.00 \\
 & PC & 0.50 & 0.00 & 0.50 & 3.00 & 3.00 & 3.00 & 0.00 & 50.00 & 50.00 & 0.00 \\
\cmidrule{2-12}
 & CCC & 0.51 & 0.00 & 0.51 & 3.04 & 3.04 & 3.04 & 0.00 & 51.00 & 49.00 & 0.00 \\
 & CGC & 0.50 & 0.01 & 0.51 & 3.02 & 3.04 & 3.04 & 1.00 & 50.00 & 49.00 & 0.00 \\
 & CPC & 0.51 & 0.00 & 0.51 & 3.04 & 3.04 & 3.04 & 0.00 & 51.00 & 49.00 & 0.00 \\
 & GCC & 0.51 & 0.00 & 0.51 & 3.04 & 3.04 & 3.04 & 0.00 & 51.00 & 49.00 & 0.00 \\
 & GGC & 0.51 & 0.00 & 0.51 & 3.04 & 3.04 & 3.04 & 0.00 & 51.00 & 49.00 & 0.00 \\
 & GGG & 0.51 & 0.00 & 0.51 & 3.04 & 3.04 & 3.04 & 0.00 & 51.00 & 49.00 & 0.00 \\
 & PCC & 0.51 & 0.00 & 0.51 & 3.04 & 3.04 & 3.04 & 0.00 & 51.00 & 49.00 & 0.00 \\
 & PPC & 0.51 & 0.00 & 0.51 & 3.04 & 3.04 & 3.04 & 0.00 & 51.00 & 49.00 & 0.00 \\
 & PPP & 0.51 & 0.00 & 0.51 & 3.04 & 3.04 & 3.04 & 0.00 & 51.00 & 49.00 & 0.00 \\
\cmidrule{2-12}
 & CCCCC & 0.51 & 0.00 & 0.51 & 3.04 & 3.04 & 3.04 & 0.00 & 51.00 & 49.00 & 0.00 \\
 & CCCCG & 0.51 & 0.00 & 0.51 & 3.04 & 3.04 & 3.04 & 0.00 & 51.00 & 49.00 & 0.00 \\
 & CCCCP & 0.51 & 0.00 & 0.51 & 3.04 & 3.04 & 3.04 & 0.00 & 51.00 & 49.00 & 0.00 \\
 & GGGGC & 0.51 & 0.00 & 0.51 & 3.04 & 3.04 & 3.04 & 0.00 & 51.00 & 49.00 & 0.00 \\
 & GGGGG & 0.51 & 0.00 & 0.51 & 3.04 & 3.04 & 3.04 & 0.00 & 51.00 & 49.00 & 0.00 \\
 & PGPGC & 0.51 & 0.00 & 0.51 & 3.04 & 3.04 & 3.04 & 0.00 & 51.00 & 49.00 & 0.00 \\
 & PPPPC & 0.51 & 0.00 & 0.51 & 3.04 & 3.04 & 3.04 & 0.00 & 51.00 & 49.00 & 0.00 \\
 & PPPPP & 0.51 & 0.00 & 0.51 & 3.04 & 3.04 & 3.04 & 0.00 & 51.00 & 49.00 & 0.00 \\
\midrule
G & None & 0.87 & 0.02 & 0.70 & 4.46 & 4.54 & 4.75 & 2.00 & 86.00 & 11.00 & 1.00 \\
\cmidrule{2-12}
 & C & 0.78 & 0.04 & 0.84 & 4.81 & 4.82 & 4.84 & 4.00 & 95.00 & 1.00 & 0.00 \\
 & G & 0.76 & 0.08 & 0.76 & 4.64 & 4.69 & 4.64 & 8.00 & 91.00 & 1.00 & 0.00 \\
 & P & 0.95 & 0.01 & 0.90 & 4.93 & 4.95 & 4.96 & 1.00 & 98.00 & 1.00 & 0.00 \\
\cmidrule{2-12}
 & CG & 0.54 & 0.30 & 0.85 & 3.74 & 4.22 & 3.73 & 30.00 & 68.00 & 1.00 & 1.00 \\
 & GC & 0.55 & 0.27 & 0.89 & 3.86 & 4.40 & 3.86 & 27.00 & 72.00 & 1.00 & 0.00 \\
 & GP & 0.61 & 0.22 & 0.78 & 3.99 & 4.14 & 4.00 & 22.00 & 74.00 & 4.00 & 0.00 \\
 & PG & 0.67 & 0.17 & 0.84 & 4.25 & 4.40 & 4.25 & 17.00 & 80.00 & 2.00 & 1.00 \\
\cmidrule{2-12}
 & CCC & 0.81 & 0.12 & 0.89 & 4.48 & 4.74 & 4.49 & 12.00 & 88.00 & 0.00 & 0.00 \\
 & CCG & 0.70 & 0.14 & 0.87 & 4.29 & 4.42 & 4.34 & 14.00 & 84.00 & 1.00 & 1.00 \\
 & CGG & 0.67 & 0.17 & 0.87 & 4.31 & 4.45 & 4.32 & 17.00 & 82.00 & 1.00 & 0.00 \\
 & GCG & 0.72 & 0.17 & 0.84 & 4.22 & 4.53 & 4.27 & 17.00 & 80.00 & 2.00 & 1.00 \\
 & GGG & 0.64 & 0.16 & 0.85 & 4.32 & 4.44 & 4.36 & 16.00 & 83.00 & 1.00 & 0.00 \\
 & GPG & 0.74 & 0.14 & 0.82 & 4.39 & 4.58 & 4.41 & 14.00 & 84.00 & 1.00 & 1.00 \\
 & PGG & 0.86 & 0.06 & 0.87 & 4.68 & 4.63 & 4.71 & 6.00 & 92.00 & 1.00 & 1.00 \\
 & PPG & 0.82 & 0.03 & 0.88 & 4.75 & 4.71 & 4.79 & 3.00 & 93.00 & 3.00 & 1.00 \\
 & PPP & 0.82 & 0.04 & 0.84 & 4.81 & 4.76 & 4.83 & 4.00 & 96.00 & 0.00 & 0.00 \\
\cmidrule{2-12}
 & CCCCC & 0.72 & 0.08 & 0.85 & 4.61 & 4.64 & 4.64 & 8.00 & 90.00 & 2.00 & 0.00 \\
 & CCCCG & 0.76 & 0.09 & 0.84 & 4.51 & 4.53 & 4.54 & 9.00 & 87.00 & 3.00 & 1.00 \\
 & CPCPG & 0.89 & 0.03 & 0.94 & 4.80 & 4.85 & 4.82 & 3.00 & 95.00 & 1.00 & 1.00 \\
 & GGGGC & 0.68 & 0.16 & 0.84 & 4.26 & 4.38 & 4.31 & 16.00 & 82.00 & 1.00 & 1.00 \\
 & GGGGG & 0.65 & 0.13 & 0.82 & 4.27 & 4.32 & 4.28 & 13.00 & 83.00 & 2.00 & 2.00 \\
 & GGGGP & 0.72 & 0.13 & 0.81 & 4.36 & 4.51 & 4.36 & 13.00 & 84.00 & 3.00 & 0.00 \\
 & PGPGC & 0.78 & 0.10 & 0.79 & 4.57 & 4.72 & 4.60 & 10.00 & 87.00 & 3.00 & 0.00 \\
 & PPPPP & 0.86 & 0.05 & 0.89 & 4.77 & 4.79 & 4.79 & 5.00 & 94.00 & 1.00 & 0.00 \\
\midrule
P & None & 0.98 & 0.00 & 0.00 & 1.05 & 1.05 & 1.05 & 0.00 & 2.00 & 95.00 & 3.00 \\
\cmidrule{2-12}
 & C & 0.10 & 0.07 & 0.04 & 1.17 & 1.17 & 1.37 & 7.00 & 5.00 & 85.00 & 3.00 \\
 & G & 0.27 & 0.09 & 0.06 & 1.14 & 1.14 & 1.22 & 9.00 & 4.00 & 85.00 & 2.00 \\
 & P & 0.37 & 0.06 & 0.02 & 1.14 & 1.14 & 1.30 & 6.00 & 4.00 & 88.00 & 2.00 \\
\cmidrule{2-12}
 & CP & 0.07 & 0.55 & 0.86 & 0.98 & 0.98 & 1.06 & 55.00 & 0.00 & 43.00 & 2.00 \\
 & GP & 0.13 & 0.53 & 0.80 & 0.98 & 0.98 & 0.98 & 53.00 & 0.00 & 45.00 & 2.00 \\
 & PC & 0.08 & 0.62 & 0.86 & 0.99 & 0.99 & 1.03 & 62.00 & 0.00 & 37.00 & 1.00 \\
 & PG & 0.08 & 0.57 & 0.87 & 0.97 & 0.97 & 0.97 & 57.00 & 0.00 & 40.00 & 3.00 \\
\cmidrule{2-12}
 & CCC & 0.04 & 0.64 & 0.88 & 0.98 & 0.98 & 0.98 & 64.00 & 0.00 & 34.00 & 2.00 \\
 & CCP & 0.09 & 0.58 & 0.81 & 0.99 & 0.99 & 0.99 & 58.00 & 0.00 & 41.00 & 1.00 \\
 & CPP & 0.08 & 0.60 & 0.86 & 0.99 & 0.99 & 0.99 & 60.00 & 0.00 & 39.00 & 1.00 \\
 & GGG & 0.17 & 0.45 & 0.61 & 0.97 & 0.97 & 1.05 & 45.00 & 0.00 & 52.00 & 3.00 \\
 & GGP & 0.05 & 0.53 & 0.75 & 1.00 & 1.00 & 1.04 & 53.00 & 0.00 & 47.00 & 0.00 \\
 & GPP & 0.11 & 0.53 & 0.72 & 0.97 & 0.97 & 1.05 & 53.00 & 0.00 & 44.00 & 3.00 \\
 & PCP & 0.08 & 0.58 & 0.87 & 0.98 & 0.98 & 0.98 & 58.00 & 0.00 & 40.00 & 2.00 \\
 & PGP & 0.04 & 0.60 & 0.88 & 1.00 & 1.00 & 1.00 & 60.00 & 0.00 & 40.00 & 0.00 \\
 & PPP & 0.14 & 0.47 & 0.71 & 0.96 & 0.96 & 0.96 & 47.00 & 0.00 & 49.00 & 4.00 \\
\cmidrule{2-12}
 & CCCCC & 0.05 & 0.56 & 0.83 & 0.98 & 0.98 & 1.02 & 56.00 & 0.00 & 42.00 & 2.00 \\
 & CCCCP & 0.07 & 0.52 & 0.79 & 1.02 & 1.02 & 1.06 & 52.00 & 1.00 & 45.00 & 2.00 \\
 & CGCGP & 0.08 & 0.56 & 0.86 & 0.99 & 0.99 & 0.99 & 56.00 & 0.00 & 43.00 & 1.00 \\
 & GGGGG & 0.10 & 0.48 & 0.67 & 0.99 & 0.99 & 1.07 & 48.00 & 0.00 & 51.00 & 1.00 \\
 & GGGGP & 0.08 & 0.58 & 0.87 & 0.97 & 0.97 & 0.97 & 58.00 & 0.00 & 39.00 & 3.00 \\
 & PPPPC & 0.11 & 0.51 & 0.78 & 1.03 & 1.03 & 1.03 & 51.00 & 1.00 & 47.00 & 1.00 \\
 & PPPPG & 0.08 & 0.51 & 0.75 & 0.96 & 0.96 & 1.00 & 51.00 & 0.00 & 45.00 & 4.00 \\
 & PPPPP & 0.10 & 0.51 & 0.80 & 0.99 & 0.99 & 0.99 & 51.00 & 0.00 & 48.00 & 1.00 \\
\bottomrule
\end{tabular}
}
\end{table*}

\begin{table*}[htbp]
\centering\small
\renewcommand{\arraystretch}{0.8}
\caption{Exhaustive Audit Summary for SWE-bench (25-Trial Sweep) with macroscopic and mechanistic degradation across all permutations. It illustrates the sycophancy threshold in medium-entropy code environments, capturing the Sovereignty Gap where models correctly derive patches but adopt a flawed swarm consensus.}
\label{tab:swe_full_results}
\resizebox{0.7\textwidth}{!}{
\begin{tabular}{lccccccccccc}
\toprule
 & & \multicolumn{3}{c}{\textbf{Macroscopic}} & \multicolumn{3}{c}{\textbf{Mechanistic}} & \multicolumn{4}{c}{\textbf{Stance (\%)}} \\
\cmidrule(lr){3-5} \cmidrule(lr){6-8} \cmidrule(lr){9-12}
\textbf{Prop.} & \textbf{Mix} & \textbf{Acc ($\mathcal{A}$)} & \textbf{Loafing ($L$)} & \textbf{Leak ($L_t$)} & \textbf{$\mathcal{E}_{cd}$} & \textbf{$\mathcal{E}_{ew}$} & \textbf{$\mathcal{E}_{ij}$} & \textbf{Adopted} & \textbf{Rejected} & \textbf{Ignored} & \textbf{Unknown} \\
\midrule
C & None & 1.00 & 0.00 & 1.00 & 5.00 & 5.00 & 5.00 & 0.00 & 100.00 & 0.00 & 0.00 \\
\cmidrule{2-12}
 & C & 1.00 & 0.00 & 1.00 & 5.00 & 5.00 & 5.00 & 0.00 & 100.00 & 0.00 & 0.00 \\
 & G & 1.00 & 0.00 & 1.00 & 5.00 & 5.00 & 5.00 & 0.00 & 100.00 & 0.00 & 0.00 \\
 & P & 1.00 & 0.00 & 1.00 & 5.00 & 5.00 & 5.00 & 0.00 & 100.00 & 0.00 & 0.00 \\
\cmidrule{2-12}
 & CG & 1.00 & 0.00 & 1.00 & 5.00 & 5.00 & 5.00 & 0.00 & 100.00 & 0.00 & 0.00 \\
 & CP & 1.00 & 0.00 & 1.00 & 5.00 & 5.00 & 5.00 & 0.00 & 100.00 & 0.00 & 0.00 \\
 & GC & 1.00 & 0.00 & 1.00 & 5.00 & 5.00 & 5.00 & 0.00 & 100.00 & 0.00 & 0.00 \\
 & PC & 1.00 & 0.00 & 1.00 & 5.00 & 5.00 & 5.00 & 0.00 & 100.00 & 0.00 & 0.00 \\
\cmidrule{2-12}
 & CCC & 1.00 & 0.00 & 1.00 & 5.00 & 5.00 & 5.00 & 0.00 & 100.00 & 0.00 & 0.00 \\
 & CGC & 1.00 & 0.00 & 1.00 & 5.00 & 5.00 & 5.00 & 0.00 & 100.00 & 0.00 & 0.00 \\
 & CPC & 1.00 & 0.00 & 1.00 & 5.00 & 5.00 & 5.00 & 0.00 & 100.00 & 0.00 & 0.00 \\
 & GCC & 1.00 & 0.00 & 1.00 & 5.00 & 5.00 & 5.00 & 0.00 & 100.00 & 0.00 & 0.00 \\
 & GGC & 1.00 & 0.00 & 1.00 & 5.00 & 5.00 & 5.00 & 0.00 & 100.00 & 0.00 & 0.00 \\
 & GGG & 1.00 & 0.00 & 1.00 & 5.00 & 5.00 & 5.00 & 0.00 & 100.00 & 0.00 & 0.00 \\
 & PCC & 1.00 & 0.00 & 1.00 & 5.00 & 5.00 & 5.00 & 0.00 & 100.00 & 0.00 & 0.00 \\
 & PPC & 1.00 & 0.00 & 1.00 & 5.00 & 5.00 & 5.00 & 0.00 & 100.00 & 0.00 & 0.00 \\
 & PPP & 1.00 & 0.00 & 1.00 & 5.00 & 5.00 & 5.00 & 0.00 & 100.00 & 0.00 & 0.00 \\
\cmidrule{2-12}
 & CCCCC & 1.00 & 0.00 & 1.00 & 5.00 & 5.00 & 5.00 & 0.00 & 100.00 & 0.00 & 0.00 \\
 & CCCCG & 1.00 & 0.00 & 1.00 & 5.00 & 5.00 & 5.00 & 0.00 & 100.00 & 0.00 & 0.00 \\
 & CCCCP & 1.00 & 0.00 & 1.00 & 5.00 & 5.00 & 5.00 & 0.00 & 100.00 & 0.00 & 0.00 \\
 & GGGGC & 1.00 & 0.00 & 1.00 & 5.00 & 5.00 & 5.00 & 0.00 & 100.00 & 0.00 & 0.00 \\
 & GGGGG & 1.00 & 0.00 & 1.00 & 5.00 & 5.00 & 5.00 & 0.00 & 100.00 & 0.00 & 0.00 \\
 & PGPGC & 1.00 & 0.00 & 1.00 & 5.00 & 5.00 & 5.00 & 0.00 & 100.00 & 0.00 & 0.00 \\
 & PPPPC & 1.00 & 0.00 & 1.00 & 5.00 & 5.00 & 5.00 & 0.00 & 100.00 & 0.00 & 0.00 \\
 & PPPPP & 1.00 & 0.00 & 1.00 & 5.00 & 5.00 & 5.00 & 0.00 & 100.00 & 0.00 & 0.00 \\
\midrule
G & None & 1.00 & 0.00 & 0.62 & 4.09 & 4.63 & 4.72 & 0.00 & 76.00 & 24.00 & 0.00 \\
\cmidrule{2-12}
 & C & 1.00 & 0.00 & 0.89 & 4.87 & 4.92 & 4.99 & 0.00 & 96.00 & 4.00 & 0.00 \\
 & G & 1.00 & 0.00 & 0.84 & 4.89 & 4.86 & 4.93 & 0.00 & 98.00 & 2.00 & 0.00 \\
 & P & 1.00 & 0.00 & 0.93 & 4.89 & 4.96 & 4.96 & 0.00 & 96.00 & 4.00 & 0.00 \\
\cmidrule{2-12}
 & CG & 0.78 & 0.22 & 0.97 & 4.10 & 4.71 & 4.15 & 22.00 & 78.00 & 0.00 & 0.00 \\
 & GC & 0.86 & 0.14 & 0.95 & 4.39 & 4.71 & 4.40 & 14.00 & 84.00 & 2.00 & 0.00 \\
 & GP & 0.80 & 0.19 & 0.94 & 4.20 & 4.67 & 4.22 & 19.00 & 80.00 & 1.00 & 0.00 \\
 & PG & 0.88 & 0.12 & 0.93 & 4.47 & 4.78 & 4.47 & 12.00 & 87.00 & 0.00 & 1.00 \\
\cmidrule{2-12}
 & CCC & 0.91 & 0.08 & 0.90 & 4.65 & 4.85 & 4.69 & 8.00 & 90.00 & 2.00 & 0.00 \\
 & CCG & 0.91 & 0.09 & 0.92 & 4.59 & 4.82 & 4.65 & 9.00 & 89.00 & 2.00 & 0.00 \\
 & CGG & 0.93 & 0.07 & 0.95 & 4.72 & 4.82 & 4.72 & 7.00 & 93.00 & 0.00 & 0.00 \\
 & GCG & 0.90 & 0.10 & 0.96 & 4.60 & 4.77 & 4.60 & 10.00 & 90.00 & 0.00 & 0.00 \\
 & GGG & 0.92 & 0.08 & 0.92 & 4.66 & 4.81 & 4.64 & 8.00 & 91.00 & 1.00 & 0.00 \\
 & GPG & 0.99 & 0.01 & 0.92 & 4.96 & 5.00 & 4.96 & 1.00 & 99.00 & 0.00 & 0.00 \\
 & PGG & 1.00 & 0.00 & 0.86 & 4.96 & 4.94 & 5.00 & 0.00 & 98.00 & 2.00 & 0.00 \\
 & PPG & 0.98 & 0.02 & 0.91 & 4.91 & 4.98 & 4.92 & 2.00 & 97.00 & 1.00 & 0.00 \\
 & PPP & 0.99 & 0.01 & 0.92 & 4.94 & 4.99 & 4.96 & 1.00 & 97.00 & 2.00 & 0.00 \\
\cmidrule{2-12}
 & CCCCC & 0.97 & 0.03 & 0.92 & 4.84 & 4.96 & 4.89 & 3.00 & 96.00 & 1.00 & 0.00 \\
 & CCCCG & 0.95 & 0.05 & 0.92 & 4.80 & 4.89 & 4.80 & 5.00 & 95.00 & 0.00 & 0.00 \\
 & CPCPG & 0.99 & 0.01 & 0.92 & 4.94 & 4.91 & 4.96 & 1.00 & 97.00 & 2.00 & 0.00 \\
 & GGGGC & 0.88 & 0.12 & 0.90 & 4.48 & 4.75 & 4.49 & 12.00 & 87.00 & 1.00 & 0.00 \\
 & GGGGG & 0.86 & 0.14 & 0.89 & 4.46 & 4.76 & 4.48 & 14.00 & 84.00 & 2.00 & 0.00 \\
 & GGGGP & 0.97 & 0.03 & 0.92 & 4.86 & 4.87 & 4.88 & 3.00 & 96.00 & 1.00 & 0.00 \\
 & PPPPG & 1.00 & 0.00 & 0.84 & 4.94 & 4.95 & 4.95 & 0.00 & 98.00 & 1.00 & 1.00 \\
 & PPPPP & 1.00 & 0.00 & 0.90 & 4.96 & 5.00 & 5.00 & 0.00 & 100.00 & 0.00 & 0.00 \\
\midrule
P & None & 1.00 & 0.00 & 0.97 & 4.84 & 4.94 & 4.91 & 0.00 & 94.00 & 6.00 & 0.00 \\
\cmidrule{2-12}
 & C & 0.95 & 0.05 & 1.00 & 4.81 & 5.00 & 4.85 & 5.00 & 95.00 & 0.00 & 0.00 \\
 & G & 0.96 & 0.04 & 0.99 & 4.84 & 5.00 & 4.84 & 4.00 & 96.00 & 0.00 & 0.00 \\
 & P & 1.00 & 0.00 & 0.99 & 5.00 & 5.00 & 5.00 & 0.00 & 100.00 & 0.00 & 0.00 \\
\cmidrule{2-12}
 & CP & 0.21 & 0.76 & 0.86 & 1.53 & 3.55 & 1.82 & 76.00 & 14.00 & 7.00 & 3.00 \\
 & GP & 0.19 & 0.79 & 0.82 & 1.48 & 3.03 & 1.91 & 79.00 & 13.00 & 4.00 & 4.00 \\
 & PC & 0.31 & 0.65 & 0.76 & 1.98 & 3.54 & 2.36 & 65.00 & 25.00 & 8.00 & 2.00 \\
 & PG & 0.22 & 0.76 & 0.89 & 1.52 & 3.48 & 1.88 & 76.00 & 14.00 & 6.00 & 4.00 \\
\cmidrule{2-12}
 & CCC & 0.35 & 0.62 & 0.79 & 1.73 & 3.47 & 2.40 & 62.00 & 19.00 & 16.00 & 3.00 \\
 & CCP & 0.22 & 0.71 & 0.91 & 1.64 & 2.84 & 1.96 & 71.00 & 16.00 & 13.00 & 0.00 \\
 & CPP & 0.22 & 0.72 & 0.90 & 1.64 & 3.26 & 2.03 & 72.00 & 16.00 & 12.00 & 0.00 \\
 & GGG & 0.40 & 0.59 & 0.72 & 1.68 & 4.02 & 2.57 & 59.00 & 17.00 & 24.00 & 0.00 \\
 & GGP & 0.34 & 0.65 & 0.75 & 1.36 & 4.16 & 2.50 & 65.00 & 9.00 & 26.00 & 0.00 \\
 & GPP & 0.27 & 0.73 & 0.78 & 1.19 & 3.97 & 2.01 & 73.00 & 6.00 & 18.00 & 3.00 \\
 & PCP & 0.40 & 0.59 & 0.71 & 1.37 & 3.88 & 2.57 & 59.00 & 11.00 & 24.00 & 6.00 \\
 & PGP & 0.23 & 0.77 & 0.84 & 1.35 & 3.68 & 2.17 & 77.00 & 9.00 & 13.00 & 1.00 \\
 & PPP & 0.23 & 0.74 & 0.84 & 1.26 & 3.36 & 1.95 & 74.00 & 7.00 & 17.00 & 2.00 \\
\cmidrule{2-12}
 & CCCCC & 0.47 & 0.52 & 0.69 & 1.98 & 3.98 & 3.05 & 52.00 & 25.00 & 20.00 & 3.00 \\
 & CCCCP & 0.32 & 0.66 & 0.79 & 1.62 & 3.56 & 2.23 & 66.00 & 16.00 & 16.00 & 2.00 \\
 & CGCGP & 0.37 & 0.56 & 0.76 & 1.84 & 3.28 & 2.31 & 56.00 & 21.00 & 21.00 & 2.00 \\
 & GGGGG & 0.48 & 0.49 & 0.66 & 1.69 & 3.64 & 2.76 & 49.00 & 18.00 & 30.00 & 3.00 \\
 & GGGGP & 0.26 & 0.70 & 0.75 & 1.11 & 3.05 & 2.04 & 70.00 & 3.00 & 26.00 & 1.00 \\
 & PPPPC & 0.39 & 0.61 & 0.67 & 1.24 & 3.56 & 2.27 & 61.00 & 7.00 & 28.00 & 4.00 \\
 & PPPPG & 0.36 & 0.62 & 0.70 & 1.23 & 4.08 & 2.37 & 62.00 & 5.00 & 31.00 & 2.00 \\
 & PPPPP & 0.29 & 0.66 & 0.79 & 1.41 & 3.32 & 2.09 & 66.00 & 11.00 & 19.00 & 4.00 \\
\bottomrule
\end{tabular}
}
\end{table*}

\clearpage

\end{document}